\documentclass{article}%
\usepackage{amsmath}
\usepackage{amsfonts}
\usepackage{amssymb}
\usepackage{graphicx}%
\providecommand{\U}[1]{\protect\rule{.1in}{.1in}}
\newtheorem{theorem}{Theorem}

\newtheorem{corollary}[theorem]{Corollary}

\newtheorem{definition}[theorem]{Definition}

\newtheorem{lemma}[theorem]{Lemma}

\newtheorem{proposition}[theorem]{Proposition}
\newtheorem{remark}[theorem]{Remark}

\newenvironment{proof}[1][Proof]{\noindent\textbf{#1.} }{\ \rule{0.5em}{0.5em}}
\newcommand{\bpartial}{\mathop{\partial\kern -4pt\raisebox{.8pt}{$|$}}}
\newcommand{\bra}{\mathopen{[\kern-1.6pt[}}
\newcommand{\ket}{\mathclose{]\kern-1.5pt]}}
\newcommand{\bbra}{\mathopen{[\kern-2.2pt[\kern-2.3pt[}}
\newcommand{\bket}{\mathclose{]\kern-2.1pt]\kern-2.3pt]}}

\makeindex
\begin{document}

\title{A Clifford Bundle Approach to the Differential Geometry of
Branes\thanks{Advances in Applied Clifford Algebras (2014) DOI
10.1007/s00006-014-0452-6}}
\author{Waldyr A. Rodrigues Jr.$^{(1)}$ and Samuel A. Wainer$^{(2)}$\\$\hspace{-0.1cm}$Institute of Mathematics, Statistics and Scientific Computation\\IMECC-UNICAMP\\13083-859 Campinas, SP, Brazil\\e-mail:$^{(1)}$ walrod@ime.unicamp.br\smallskip\ and $^{(2)}$samuelwainer@ime.unicamp.br}
\date{March 25 2014}
\maketitle

\begin{abstract}
We first recall using the Clifford bundle formalism (CBF) of differential
forms and the theory of extensors acting on $\mathcal{C\ell}(M,g)$ (the
Clifford bundle of differential forms) the formulation of the intrinsic
geometry of a differential manifold $M$ equipped with a metric field
$\boldsymbol{g}$ of signature $(p,q)$ and an arbitrary metric compatible
connection $\nabla$ introducing the torsion ($2-1$)-extensor field $\tau$, the
curvature $(2-2)$ extensor field $\mathfrak{R}$ and (once fixing a gauge) the
connection $(1-2)$-extensor $\omega$ and the Ricci operator
$\boldsymbol{\partial}\wedge\boldsymbol{\partial}$ (where
$\boldsymbol{\partial}$ is the Dirac operator acting on sections of
$\mathcal{C\ell}(M,g)$) which plays an important role in this paper. Next,
using the CBF we give a thoughtful presentation the Riemann or the Lorentzian
geometry of an orientable submanifold $M$ ($\dim M=m$) living in a manifold
$\mathring{M}$ (such that $\mathring{M}\simeq\mathbb{R}^{n}$ is equipped with
a semi-Riemannian metric $\boldsymbol{\mathring{g}}$ with signature
$(\mathring{p},\mathring{q})$ and \ $\mathring{p}+\mathring{q}=n$ and its
Levi-Civita connection $\mathring{D}$) and where there is defined a metric
$\boldsymbol{g=i}^{\ast}\mathring{g}$, where $\boldsymbol{i}:$ $M\rightarrow
\mathring{M}$ is the inclusion map. We prove several equivalent forms for the
curvature operator $\mathfrak{R}$ of $M$. Moreover we show a very important
result, namely that the Ricci operator of $M$ is the (negative) square of the
shape operator $\mathbf{S}$ of $M$ (object obtained by applying the
restriction on $M$ of the Dirac operator $\boldsymbol{\mathring{\partial}}$ of
$\mathcal{C\ell}(\mathring{M},\mathring{g})$ to the projection operator
$\mathbf{P}$). Also we disclose the relationship between the ($1-2$)-extensor
$\omega$ and the shape biform $\mathcal{S}$ (an object related to $\mathbf{S}%
$). The results obtained are used to give a mathematical formulation to
Clifford's theory of matter. It is hoped that our presentation will be useful
for differential geometers and theoretical physicists interested, e.g., in
string and brane theories and relativity theory\ by divulging, improving and
expanding very important and so far unfortunately largely ignored results
appearing in reference \cite{hs1984}.\newpage

\end{abstract}
\tableofcontents

\section{Introduction}

In this paper we use the Clifford bundle formalism (CBF) in order to analyze
the Riemann or the Lorentzian geometry of an orientable submanifold $M$ ($\dim
M=m$) living in a manifold $\mathring{M}$ such that $\mathring{M}%
\simeq\mathbb{R}^{n}$ is equipped with a semi-Riemannian metric
$\boldsymbol{\mathring{g}}$ (with signature $(\mathring{p},\mathring{q})$
and\ $\mathring{p}+\mathring{q}=n$) and its Levi-Civita connection
$\mathring{D}$.

In order to achieve our objectives and exhibit some nice results that are not
well known (and which, e.g., may possibly be of interest for the description
and formulation of branes theories \cite{ma} and string theories
\cite{becker}) we first recall in Section 2 how to formulate using the CBF the
intrinsic differential geometry of a structure $\langle M,\boldsymbol{g}%
,\nabla\rangle$ where $\nabla$ is a general metric compatible Riemann-Cartan
connection, i.e., $\nabla\boldsymbol{g=0}$ and the Riemann and torsion tensors
of $\nabla$ are non null. In our approach we will introduce (once we fix a
gauge in the frame bundle) a $(1,2)$-extensor field $\boldsymbol{\omega}:\sec%
{\textstyle\bigwedge\nolimits^{1}}
T^{\ast}M\rightarrow%
{\textstyle\bigwedge\nolimits^{2}}
T^{\ast}M$ closed related with the connection $1$-forms which permits to write
a very nice formula for the covariant derivative (see Eq.(\ref{code})) of any
section of the Clifford bundle of the structure $\langle M,\boldsymbol{g}%
,\nabla\rangle$. It will be shown that $\boldsymbol{\omega}$ is related to
$\mathcal{S}:\sec%
{\textstyle\bigwedge\nolimits^{1}}
T^{\ast}M\rightarrow%
{\textstyle\bigwedge\nolimits^{2}}
T^{\ast}M$ the shape operator biform of the manifold.

Then in Section 3, we suppose that $M$ is a \emph{proper}
submanifold\footnote{By a proper (or regular) submanifold $M$ of $\mathring
{M}$ we mean a subset $M$ $\subset$ $\mathring{M}$ such that for every $x\in
M$ in the domain of a chart \ $(U,\sigma)$ of\ $\mathring{M}$ such that
$\sigma:$ $\mathring{M}\cap U\rightarrow\mathbb{R}^{n}\times\{\mathbf{l}%
\},\newline\sigma(x)=(x^{1},\cdots,x^{n},l^{1},\cdots,l^{m-n})$, where
$\mathbf{l\in}\mathbb{R}^{n-m}$.} of $\mathring{M}$ which $\boldsymbol{i}%
:M\mapsto\mathring{M}$ the inclusion map. Introducing natural \emph{global}
coordinates $(\boldsymbol{x}^{1},...,\boldsymbol{x}^{n})$\ for $\mathring
{M}\simeq\mathbb{R}^{n}$ we write $\boldsymbol{\mathring{g}=}%
{\textstyle\sum\nolimits_{i,j=1}^{n}}
\eta_{ij}d\boldsymbol{x}^{i}\otimes d\boldsymbol{x}^{j}$ $\equiv\eta
_{ij}d\boldsymbol{x}^{i}\otimes d\boldsymbol{x}^{j}$ and equip $\mathring{M}$
with the pullback metric $\boldsymbol{g:=i}^{\ast}\boldsymbol{\mathring{g}}$.
We then find the relation between the Levi-Civita connection $D$ of
$\boldsymbol{g}$ and $\mathring{D}$, the Levi-Civita connection of
$\boldsymbol{\mathring{g}}$. We suppose that $\boldsymbol{g}$ is non
degenerated of signature $(p,q)$ with \ $p+q=m$.

$\mathcal{C\ell}(\mathring{M},\mathtt{\mathring{g}})$ and $\mathcal{C\ell
}(M,\mathtt{g})$\ denote respectively the Clifford bundles of differential
forms of $\mathring{M}$ and $M$\footnote{For all applications in what follows
take notice that $%
{\textstyle\bigwedge}
T^{\ast}M=%
{\textstyle\bigoplus\nolimits_{r=0}^{n}}
{\textstyle\bigwedge^{r}}
T^{\ast}M$ $\hookrightarrow\mathcal{C\ell}(M,\mathtt{g})$, where the symbol
$\hookrightarrow$ means that \ for each $x\in M$, $%
{\textstyle\bigwedge}
T_{x}^{\ast}M$ (the bundle of differential forms) is embedded in
$\mathcal{C\ell}(T_{x}^{\ast}M,\mathtt{g}_{x})$ and $%
{\textstyle\bigwedge}
T_{x}^{\ast}M\subseteq\mathcal{C\ell}(%
{\textstyle\bigwedge}
T_{x}^{\ast}M,\mathtt{g}_{x})$.}. Moreover, in what follows $\mathring{g}=%
{\textstyle\sum\nolimits_{i,j=1}^{n}}
\eta^{ij}\frac{\partial}{\partial\boldsymbol{x}^{i}}\otimes\frac{\partial
}{\partial\boldsymbol{x}^{j}}\equiv\eta^{ij}\frac{\partial}{\partial
\boldsymbol{x}^{i}}\otimes\frac{\partial}{\partial\boldsymbol{x}^{j}}$ is the
metric of the cotangent bundle. The Dirac operators\footnote{Take notice that
the Dirac operators used in this paper is acting on sections of the Clifford
bundle. It is not to be confused with the Dirac operator which acts on
sections of the spinor bundle (see details in \cite{nra}) . This last operator
can be used to probe the topology of the brane, as showed in \cite{rbh}} of
$\mathcal{C\ell}(\mathring{M},\mathtt{\mathring{g}})$\ and $\mathcal{C\ell
}(M,\mathtt{g})$ will be denoted\footnote{We follow here, whenever possible,
the notation used in \cite{rodcap2007}. Note that differently from references
\cite{hs1984,h1986,sobczyk} we use the left and right contractions operators
$\lrcorner$ and $\llcorner$ and the scalar product operator (denoted by
$\cdot$ ) acting on sections of the Clifford bundle. Also our convention for
the Riemann tensor makes some equations to have a different signals than ones
appearing in the references just quoted.}\ by $\boldsymbol{\mathring{\partial
}}$ and $\boldsymbol{\partial}$. Let $l=n-m$ and $\{\mathring{e}_{1}%
,\mathring{e}_{2},...,\mathring{e}_{\boldsymbol{m}},\mathring{e}%
_{m+1},...,\mathring{e}_{m+l}\}$ an orthonormal basis for $T\mathring{U}$
($\mathring{U}\subset\mathring{M}$) such that $\{e_{\mathbf{1}},\boldsymbol{e}%
_{\mathbf{2}},...,e_{\mathbf{m}}\}=\{\mathring{e}_{1},\mathring{e}%
_{2},...,\mathring{e}_{\boldsymbol{m}}\}$ is a basis for $TU$ ($U\subset
\mathring{U}$) \ and if $\{\mathring{\theta}^{\mathbf{1}},\mathring{\theta
}^{\mathbf{2}},...,\mathring{\theta}^{\mathbf{m}},\mathring{\theta}%
^{m+1},...,\mathring{\theta}^{m+l}\}$ \ is the dual basis of the $\{e_{i}\}$
we have that $\{\theta^{\mathbf{1}},\theta^{\mathbf{2}},...,\theta
^{\mathbf{m}}\}=\{\mathring{\theta}^{1},\mathring{\theta}^{2},...,\mathring
{\theta}^{m}\}$ is a basis for $T^{\ast}U$ dual to the basis $\{e_{\mathbf{1}%
},\boldsymbol{e}_{\mathbf{2}},...,e_{\mathbf{m}}\}$ of $TU$. We have, as well
known \cite{rodcap2007}:
\begin{equation}
\boldsymbol{\mathring{\partial}}=%
{\textstyle\sum\nolimits_{i=1}^{n}}
\mathring{\theta}^{i}\mathring{D}_{\boldsymbol{e}_{i}}=\mathring{\theta}%
^{i}\mathring{D}_{\boldsymbol{e}_{i}}\text{, \ }\boldsymbol{\partial
=}\text{\ }%
{\textstyle\sum\nolimits_{i=1}^{m}}
\theta^{i}D_{\boldsymbol{e}_{i}}=\theta^{\mathbf{i}}D_{\boldsymbol{e}%
_{\mathbf{i}}}, \label{1}%
\end{equation}

\begin{remark}
Take notice tha the bold face sub and superscripts are use to denote bases
$\{e_{\mathbf{i}}\}$ and $\{\theta^{\mathbf{i}}\}$ of the tangent and
cotangent space of $M$. This notation is conveniently used in what follows.
\end{remark}

The dual basis of the natural coordinate basis $\{\frac{\partial}%
{\partial\boldsymbol{x}^{i}}\}$ is denoted in what follows by $\{\gamma^{i}\}$
where, of course, $\gamma^{i}=d\boldsymbol{x}^{i}$. Moreover, we denote by
$\{\mathring{e}^{1},\mathring{e}^{2},...,\mathring{e}^{m}\}$ the reciprocal
frame of \ $\{\mathring{e}_{i}\}$, i.e., $\boldsymbol{\mathring{g}}%
(\mathring{e}^{i},\mathring{e}_{j})=\delta_{j}^{i}$ and by $\{\mathring
{\theta}_{i}\}$ the reciprocal basis of$\{\mathring{\theta}^{i}\}$, i.e.,
$\mathring{g}(\mathring{\theta}^{i},\mathring{\theta}_{j}):=\mathring{\theta
}^{i}\cdot\mathring{\theta}_{j}=\delta_{j}^{i}$ Moreover, take into account
that for $\mathbf{i},\mathbf{j}=1,...,m$ it is $\mathtt{g}(\theta
^{^{\mathbf{i}}},\theta_{\mathbf{j}})=\mathtt{\mathring{g}}(\mathring{\theta
}^{i},\mathring{\theta}_{j})$ So we will write also $\mathtt{g}(\theta
^{\mathbf{i}},\theta_{\mathbf{j}})=$ $\theta^{\mathbf{i}}\cdot\theta
_{\mathbf{j}}=\delta_{\mathbf{j}}^{\mathbf{i}}$. The representation of the
Dirac operator $\boldsymbol{\mathring{\partial}}$ in the natural coordinate
basis of $\mathring{M}$ is of course, $\mathfrak{%
{\textstyle\sum\nolimits_{i=1}^{n}}
\gamma}^{i}\frac{\partial}{\partial\boldsymbol{x}^{i}}=%
{\textstyle\sum\nolimits_{i=1}^{n}}
\mathring{\theta}^{i}\mathring{D}_{\boldsymbol{e}_{i}}$. Note that we have
$\left.  \mathring{\theta}^{m+1}\right\vert _{M}=0,...,\left.  \mathring
{\theta}^{m+l}\right\vert _{M}=0$, i.e., the $\{\mathring{\theta}%
^{m+1},...,\mathring{\theta}^{m+l}\}$ for any vector field $\boldsymbol{a}%
\in\sec TU$ and $d=1,...,m+l$ we have%
\[
\left.  \mathring{\theta}^{m+d}\right\vert _{M}(\boldsymbol{a})=0.
\]
We denote moreover
\begin{equation}
\mathfrak{\mathring{d}=}\left.  \boldsymbol{\mathring{\partial}}\right\vert
_{M}:=\theta^{\mathbf{i}}\theta_{\mathbf{i}}\cdot\boldsymbol{\mathring
{\partial}=}%
{\textstyle\sum\nolimits_{i=1}^{m}}
\theta^{\mathbf{i}}\mathring{D}_{\boldsymbol{e}_{\mathbf{i}}}=\theta
^{\mathbf{i}}\mathring{D}_{\boldsymbol{e}_{i}} \label{2a}%
\end{equation}
the restriction of $\boldsymbol{\mathring{\partial}}$ on the submanifold $M$.
The projection operator $\mathbf{P}$ (an extensor field\footnote{For a
thougtful presentation of the theory of extensor fields, see, e.g.,
\cite{rodcap2007}.}) on $M$ and the \emph{shape} \emph{operator }%
$\mathbf{S}=\mathfrak{\mathring{d}{}}\mathbf{P}$\textbf{:} $\sec
\mathcal{C\ell}(\mathring{M},\mathtt{\mathring{g}})\rightarrow\sec
\mathcal{C\ell}(M,\mathtt{g})$ and \emph{shape biform} operator of the
manifold $M$, $\mathcal{S}:$ $\sec%
{\textstyle\bigwedge\nolimits^{1}}
T^{\ast}M\mapsto\
{\textstyle\bigwedge\nolimits^{2}}
T^{\ast}M,$ $\mathcal{S}(a):=-(a\cdot\mathfrak{d}I_{m})I_{m}^{-1}$ (where
$\tau_{\boldsymbol{g}}=I_{m}=\theta^{1}\theta^{2}\cdots\theta^{m}$ is the
volume form\footnote{The volume $\tau_{\boldsymbol{\mathring{g}}}$ for on
$\mathring{U}\subset\mathring{M}$ $\ $will be denoted by $I_{n}%
=\boldsymbol{\mathring{\theta}}^{1}\boldsymbol{\mathring{\theta}}^{2}%
\cdots\boldsymbol{\mathring{\theta}}^{m}$. The volume form $\tau
_{\boldsymbol{\mathring{g}}}$ on $\mathring{U}\subset M$ will be denoted
$I_{n}=\mathring{\theta}^{\mathbf{1}}\mathring{\theta}^{\mathbf{2}}%
\cdots\mathring{\theta}^{\mathbf{m}}\mathring{\theta}^{m+1}\cdots
\mathring{\theta}^{m+l}=I_{m}\mathring{\theta}^{\mathbf{m}}\mathring{\theta
}^{m+1}\cdots\mathring{\theta}^{m+l}$.} on $U\subset M$ ) are fundamental
objects in this study. The definition of those objects are given in Section 3
and the main algebraic properties of $\mathbf{P,}$ $\mathbf{S}$ and
$\mathcal{S}$\ besides all identities necessary for the present paper are
given and proved at the appropriate places.

Section 4 is dedicated to find several equivalent expressions for the
curvature biform $\mathfrak{R}(u,v)$ in terms of the shape operator. There we
recall that the square of the Dirac operator $\boldsymbol{\partial}$\ acting
on sections of the Clifford bundle has two different decompositions, namely%
\begin{equation}
\boldsymbol{\partial}^{2}=-(d\delta+\delta d)=\boldsymbol{\partial
\cdot\partial+\partial\wedge\partial,} \label{3}%
\end{equation}
where $d$ and $\delta$ are respectively the exterior derivative and the\ Hodge
coderivative\ and $\boldsymbol{\partial\cdot\partial+\partial\wedge\partial}$
are respectively the covariant Laplacian and the Ricci operator. The explicit
forms of $\boldsymbol{\partial\cdot\partial}$ and $\boldsymbol{\partial
\wedge\partial}$ are given in \cite{rodcap2007} where it is shown moreover
that $\boldsymbol{\partial\wedge\partial}$ is an extensorial operator and the
remarkable result
\begin{equation}
\boldsymbol{\partial\wedge\partial~}\theta^{i}=\mathcal{R}^{i}, \label{3a}%
\end{equation}
where the objects $\mathcal{R}^{i}=R_{j}^{i}\theta^{j}\in\sec%
{\textstyle\bigwedge\nolimits^{1}}
T^{\ast}M\hookrightarrow\sec\mathcal{C\ell}(\mathring{M},\mathtt{\mathring{g}%
})$ with $R_{j}^{i}$ the components of the Ricci tensor associated with $D$
are called the Ricci $1$-form fields. One of the main purposes of the present
paper is to give (Section 5) a detailed proof of the\ remarkable equation%
\begin{equation}
\boldsymbol{\partial\wedge\partial~(}v)=-\mathbf{S}^{2}(v), \label{3aa}%
\end{equation}
which says that the shape biform operator is the negative square root of the
Ricci operator\footnote{This result appears (with a positive sign on the
second member of Eq.(\ref{3aa}) in \cite{hs1984}. See also \cite{sobczyk}.
However, take into account that the methods used in those references use the
Clifford algebra of multivectors and thus, comparison of the results there
with the standard presentations of modern differential geometry using
differential forms are not so obvious, this being probably one of the reasons
why some important and beautiful results displayed in \cite{hs1984} are
unfortunately ignored.}. We moreover find the relation between $\mathcal{S}%
(v)$ and $\boldsymbol{\omega}(v)$ thus providing a very interesting
geometrical meaning for the connection $1$-forms $\omega_{\cdot j}^{i\cdot}$
of the Levi-Civita connection $D$, namely as the\ angular `velocity' with
which the pseudo scalar $I_{m}$ when it slides on $M$.

We also discuss in Section 5 if the present formalism permits to give a
mathematical representation concerning Clifford's space theory of matter. In
Section 6 we show that in a Lorentzian brane containing a Killing vector field
Einstein equation can be encoded in a Maxwell like equation whose source is a
current given by $J=2\mathbf{S}^{2}(A)$. The article contains also an Appendix
presenting some identities involving the projection operator and its covariant
derivative which permit to prove Proposition \ref{prodif}.

In Section 7 we present our conclusions.

\section{Curvature and Torsion Extensor of a Riemann-Cartan Connection}

Let $\boldsymbol{u},\boldsymbol{v},\boldsymbol{t},\boldsymbol{z}\in\sec TM$
and $u,v,t,z$ $\in\sec\bigwedge^{1}T^{\ast}M\hookrightarrow\sec\mathcal{C\ell
}(M,\mathtt{g})$ the physically equivalent $1$-forms, i.e., $u=\boldsymbol{g}%
(\boldsymbol{u}\mathbf{,})$, etc.\textbf{ }Let moreover $\{\boldsymbol{e}%
_{\mathbf{a}}\}$ be an orthonormal basis for $TM$ and $\{\theta^{\mathbf{a}%
}\}$, $\theta^{\mathbf{a}}\in\sec\bigwedge\nolimits^{1}T^{\ast}%
M\hookrightarrow\mathcal{C\ell}(M,\mathtt{g})$ the corresponding dual basis
and consider the Riemann-Cartan structure $(M,\boldsymbol{g},\nabla)$.

\begin{definition}
The form derivative of $M$ is the operator
\begin{gather}
\eth:\sec\mathcal{C\ell}(M,\mathtt{g})\rightarrow\sec\mathcal{C\ell
}(M,\mathtt{g}),\nonumber\\
\eth\mathfrak{\mathcal{C}:}=\theta^{\mathbf{a}}\eth_{\boldsymbol{e}%
_{\mathbf{a}}}\mathfrak{\mathcal{C}} \label{formderiv1}%
\end{gather}
where $\mathfrak{d}_{\boldsymbol{e}_{\mathbf{a}}}$ is the Pfaff derivative of
form fields
\begin{equation}
\eth_{\boldsymbol{e}_{\mathbf{a}}}\mathfrak{\mathcal{C}}:\mathfrak{\mathcal{=}%
}%
{\textstyle\sum\nolimits_{r=0}^{m}}
\eth_{\boldsymbol{e}_{\mathbf{a}}}\langle\mathfrak{\mathcal{C}}\rangle_{r}
\label{pffaf0}%
\end{equation}
such that if $\langle\mathfrak{\mathcal{C}}\rangle_{r}$ is expanded in the
basis generated by $\{\theta^{\mathbf{a}}\}$, i.e., $\langle
\mathfrak{\mathcal{C}}\rangle_{r}=\mathfrak{\mathcal{C}}_{r}=\frac{1}%
{r!}\mathcal{C}_{\mathbf{i}_{1}\cdots\mathbf{i}_{r}}\theta^{\mathbf{i}%
_{1}\cdots\mathbf{i}_{r}}\in\sec\bigwedge\nolimits^{r}T^{\ast}M\hookrightarrow
\sec\mathcal{C\ell}(M,\mathtt{g})$ it is%
\begin{equation}
\eth_{\boldsymbol{e}_{\mathbf{a}}}\langle\mathfrak{\mathcal{C}}\rangle
_{r}:=\frac{1}{r!}\boldsymbol{e}_{\mathbf{a}}(\mathcal{C}_{\mathbf{i}%
_{1}\cdots\mathbf{i}_{r}}\theta^{\mathbf{i}_{1}\cdots\mathbf{i}_{r}})=\frac
{1}{r!}\boldsymbol{e}_{\mathbf{a}}(\mathcal{C}_{\mathbf{i}_{1}\cdots
\mathbf{i}_{r}})\theta^{\mathbf{i}_{1}\cdots\mathbf{i}_{r}}. \label{pffaf}%
\end{equation}

\end{definition}

Given two different\ pairs of basis $\{\boldsymbol{e}_{\mathbf{a}}%
,\theta^{\mathbf{a}}\}$ \ and $\{\boldsymbol{e}_{\mathbf{a}}^{\prime}%
,\theta^{\prime\mathbf{a}}\}$ we have that
\begin{equation}
\theta^{\mathbf{a}}\eth_{\boldsymbol{e}_{\mathbf{a}}}\mathfrak{\mathcal{C}%
}=\theta^{\prime\mathbf{a}}\eth_{\boldsymbol{e}_{\mathbf{a}}}^{\prime
}\mathfrak{\mathcal{C}}, \label{pffaf1}%
\end{equation}
since for all $\mathfrak{\mathcal{C}}_{r}$
\begin{equation}
\eth^{\prime}\mathfrak{\mathcal{C}}_{r}=\theta^{\prime\mathbf{a}}%
\eth_{\boldsymbol{e}_{\mathbf{a}}}^{\prime}\mathfrak{\mathcal{C}}_{r}%
=\theta^{\prime\mathbf{a}}\boldsymbol{e}_{\mathbf{a}}^{\prime}(\frac{1}%
{r!}\mathcal{C}_{\mathbf{i}_{1}\cdots\mathbf{i}_{r}}^{\prime}\theta
^{\prime\mathbf{i}_{1}\cdots\mathbf{i}_{r}})=\theta^{\mathbf{a}}%
\boldsymbol{e}_{\mathbf{a}}(\frac{1}{r!}\mathcal{C}_{\mathbf{i}_{1}%
\cdots\mathbf{i}_{r}}\theta^{\mathbf{i}_{1}\cdots\mathbf{i}_{r}}).
\label{pffaf2}%
\end{equation}

\begin{remark}
We recall also that any biform $B\in\sec%
{\textstyle\bigwedge\nolimits^{2}}
T^{\ast}M\hookrightarrow\sec\mathcal{C}\ell(M,\mathtt{g})$ and any $A_{r}\in%
{\textstyle\bigwedge\nolimits^{r}}
T^{\ast}M\hookrightarrow\sec\mathcal{C}\ell(M,\mathtt{g})$ with $r\geq2$ it
holds that
\begin{equation}
BA_{r}=B\lrcorner A_{r}+B\times A_{r}+B\wedge A_{r}\text{.} \label{pcomu}%
\end{equation}
where for any $\mathcal{C},\mathcal{D}\in\sec\mathcal{C}\ell(M,\mathtt{g})$
\begin{equation}
\mathcal{C}\times\mathcal{D}=\frac{1}{2}(\mathcal{CD-DC}) \label{pcomu1}%
\end{equation}
We observe that \ for $v\in%
{\textstyle\bigwedge\nolimits^{1}}
T^{\ast}M\hookrightarrow\sec\mathcal{C}\ell(M,\mathtt{g})$ it is%
\begin{equation}
B\times v=B\llcorner v=-v\lrcorner B. \label{pcomu2}%
\end{equation}

\end{remark}

Call $\overset{\triangledown}{\mathbf{\partial}}$ $:=\theta^{\mathbf{a}}%
\nabla_{\boldsymbol{e}_{\mathbf{a}}}$ the Dirac operator associated with
$\nabla$, a general Riemann-Cartan connection. In \cite{rodcap2007} it is
introduced the Dirac commutator\emph{ }of two $1$-form fields $u,v\in
\sec\bigwedge\nolimits^{1}T^{\ast}M\hookrightarrow\sec\mathcal{C\ell
}(M,\mathtt{g})$ associated with $\nabla$ by%
\begin{align*}%
\bbra
\;,\;%
\bket
&  :\sec\bigwedge\nolimits^{1}T^{\ast}M\times\text{ }\sec\bigwedge
\nolimits^{1}T^{\ast}M\rightarrow\sec\bigwedge\nolimits^{1}T^{\ast}M\\%
\bbra
u,v%
\bket
&  =(u\cdot\overset{\triangledown}{\mathbf{\partial}})v-(v\cdot\overset
{\triangledown}{\mathbf{\partial}})u-%
\bra
u,v%
\ket
\end{align*}
where
\begin{equation}%
\bra
u,v%
\ket
=(u\cdot\boldsymbol{\partial})v-(v\cdot\boldsymbol{\partial})u, \label{ten1}%
\end{equation}
is the Lie bracket of $1$-form fields\footnote{We have, e.g., that if
$[\boldsymbol{e}_{\mathbf{a}},\boldsymbol{e}_{\mathbf{b}}]=c_{\cdot
\mathbf{ab}}^{\mathbf{d\cdot\cdot}}\boldsymbol{e}_{\mathbf{d}}$, then $%
\bra
\theta_{\mathbf{a}},\theta_{\mathbf{b}}%
\ket
=c_{\cdot\mathbf{ab}}^{\mathbf{d\cdot\cdot}}\theta_{\mathbf{d}}$.}.

\begin{definition}
For a metric compatible connection $\nabla$, recalling the definition of the
torsion operator\footnote{The torsion operator of a connection $\nabla$ is the
mapping $\mathbf{\tau}:\sec TM\times\sec TM\rightarrow\sec TM$,
$(\boldsymbol{u,v})\mapsto\mathbf{\tau}(\boldsymbol{u,v})=\nabla
_{\boldsymbol{u}}\boldsymbol{v}-\nabla_{v}\boldsymbol{v}-[\boldsymbol{u,v}]$.}
we conveniently write%
\begin{equation}
\mathbf{\tau(}u,v)=%
\bbra
u,v%
\bket
, \label{ten1aa}%
\end{equation}
which we call the \emph{(}form\emph{)} torsion operator.
\end{definition}

\begin{remark}
We recall\ the action of the operator\footnote{More details on the concept of
a general derivative operator $\partial_{A}$ \ ($A\in\sec\mathcal{C\ell
}(M,\mathtt{g})$) acting on a a general multiform field $E:\sec\mathcal{C\ell
}(M,\mathtt{g})\rightarrow\sec\mathcal{C\ell}(M,\mathtt{g})$ may be found,
e.g., in \cite{rodcap2007} where several explicit examples are given.}
$\partial_{u}$ \emph{(}$u\in\sec\bigwedge\nolimits^{1}T^{\ast}M\rightarrow
\sec\mathcal{C\ell}(M,\mathtt{g})$\emph{)} acting on an extensor field
$F:\sec\bigwedge\nolimits^{1}T^{\ast}M\rightarrow\sec\bigwedge\nolimits^{r}%
T^{\ast}M$, $u\mapsto F(u)$. If $u=u^{\mathbf{i}}\theta_{\mathbf{i}}$,
$\partial_{u}:=\theta^{\mathbf{k}}\frac{\partial}{\partial u^{\mathbf{k}}}$
acting on $F(u)$ is given by
\begin{align}
\partial_{u}F(u)  &  :=\theta^{\mathbf{k}}\frac{\partial}{\partial
u^{\mathbf{k}}}F(u^{\mathbf{i}}\theta_{\mathbf{i}}):=\theta^{\mathbf{k}}%
\frac{\partial}{\partial u^{\mathbf{k}}}u^{\mathbf{i}}F(\theta_{\mathbf{i}%
})\nonumber\\
&  =\theta^{\mathbf{k}}F(\theta_{\mathbf{k}})=\theta^{\mathbf{k}}\lrcorner
F(\theta_{\mathbf{k}})+\theta^{\mathbf{k}}\wedge F(\theta_{\mathbf{k}}).
\label{du1}%
\end{align}
Also the action of the operator $\partial_{u}\wedge\partial_{v}$
\emph{(}$u=u^{\mathbf{i}}\theta_{\mathbf{i}}$, $v=v^{\mathbf{i}}%
\theta_{\mathbf{i}}$\emph{)} acting on an extensor field \ $G:\sec
\bigwedge\nolimits^{1}T^{\ast}M\times\sec\bigwedge\nolimits^{1}T^{\ast
}M\hookrightarrow\sec\bigwedge\nolimits^{r}T^{\ast}M$,$\ (u,v)\mapsto G(u,v)$
is given by
\begin{align}
\partial_{u}\wedge\partial_{v}G(u,v)  &  =\theta^{\mathbf{k}}\frac{\partial
}{\partial u^{\mathbf{k}}}\wedge\theta^{\mathbf{l}}\frac{\partial}{\partial
u^{\mathbf{l}}}u^{\mathbf{m}}u^{n}G(\theta^{\mathbf{m}},\theta^{\mathbf{n}%
})\nonumber\\
&  =\theta^{\mathbf{k}}\wedge\theta^{\mathbf{l}}G(\theta_{\mathbf{k}}%
,\theta_{\mathbf{l}}). \label{du2}%
\end{align}

\end{remark}

\begin{definition}
\label{torsion x dirac copy(1)}The mapping
\begin{align}
\mathfrak{t}\text{ }  &  :\sec\bigwedge\nolimits^{2}T^{\ast}M\rightarrow
\sec\bigwedge\nolimits^{1}T^{\ast}M,\nonumber\\
\mathfrak{t}\mathbf{(}B)  &  =\frac{1}{2}B\cdot(\partial_{u}\wedge\partial
_{v})\mathbf{\tau(}u,v). \label{ten2}%
\end{align}
is called the $(2$-$1)$-extensorial torsion field and
\begin{equation}
\mathfrak{t}(u\wedge v)=\mathbf{\tau(}u,v). \label{ten2aa}%
\end{equation}

\end{definition}

\noindent Indeed, from Eq.(\ref{ten2}) we have taking $B=a\wedge b$
\begin{equation}
\mathfrak{t}\mathbf{(}a\wedge b)=\frac{1}{2}(a\wedge b)\cdot(\partial
_{u}\wedge\partial_{v})\mathbf{\tau(}u,v). \label{t1}%
\end{equation}
Now,%
\begin{equation}
\left(  \partial_{u}\wedge\partial_{v}\right)  \tau\mathbf{(}u,v)=(\theta
^{\mathbf{k}}\wedge\theta^{\mathbf{l}})\mathbf{\tau(}\theta_{\mathbf{k}%
},\theta_{\mathbf{l}}). \label{t2}%
\end{equation}
Then,%
\[
\mathfrak{t}\mathbf{(}a\wedge b)=\frac{1}{2}(a\wedge b)\cdot(\theta
^{\mathbf{k}}\wedge\theta^{\mathbf{l}})\mathbf{\tau(}\theta_{\mathbf{k}%
},\theta_{\mathbf{l}})=\tau\mathbf{(}a,b).
\]

\begin{definition}
The extensor mapping
\begin{align}
\mathbf{\Theta}  &  :\sec\bigwedge\nolimits^{1}T^{\ast}M\rightarrow
\sec\bigwedge\nolimits^{2}T^{\ast}M,\nonumber\\
\mathbf{\Theta}(c)  &  =\frac{1}{2}(\partial_{u}\wedge\partial_{v}%
)\tau\mathbf{(}u,v)\cdot c, \label{ten2a}%
\end{align}
is called the Cartan torsion field.
\end{definition}

We have that
\[
\mathfrak{t}(u\wedge v)=\partial_{c}(u\wedge v)\cdot\mathbf{\Theta}(c)
\]
and if $\nabla_{\boldsymbol{e}_{\mathbf{a}}}\theta^{\mathbf{b}}:=-\omega
_{\cdot\mathbf{ac}}^{\mathbf{b\cdot\cdot}}\theta^{\mathbf{c}}$ then
\begin{align}
z\cdot\mathfrak{t}(u\wedge v)  &  =z_{\mathbf{d}}u^{\mathbf{a}}v^{\mathbf{b}%
}T_{\cdot\mathbf{ab}}^{\mathbf{d\cdot\cdot}},\nonumber\\
T_{\mathbf{\cdot ab}}^{\mathbf{c\cdot\cdot}}  &  =\omega_{\mathbf{\cdot ab}%
}^{\mathbf{c\cdot\cdot}}-\omega_{\mathbf{\cdot ba}}^{\mathbf{c\cdot\cdot}%
}-c_{\mathbf{\cdot ab}}^{\mathbf{c\cdot\cdot}}\text{ }. \label{TORSION rie}%
\end{align}

\begin{definition}
The connection $(1$-$2)$-extensor field $\boldsymbol{\omega}$ in a given gauge
is given by \emph{(}$v=\boldsymbol{g}(\boldsymbol{v}\mathbf{,}$ $)$\emph{)}
\begin{align}
\omega &  :\sec\bigwedge\nolimits^{1}T^{\ast}M\rightarrow\sec\bigwedge
\nolimits^{2}T^{\ast}M,\nonumber\\
v  &  \mapsto\omega(v)=\frac{1}{2}v^{\mathbf{c}}\omega_{\cdot\mathbf{c}\cdot
}^{\mathbf{a}\cdot\mathbf{b}}\theta_{\mathbf{a}}\wedge\theta_{\mathbf{b}}.
\label{om1}%
\end{align}

\end{definition}

We also introduce the operator
\begin{align}
\boldsymbol{\omega}  &  :\sec\bigwedge\nolimits^{1}TM\rightarrow\sec
\bigwedge\nolimits^{2}T^{\ast}M,\nonumber\\
\boldsymbol{v}  &  \mapsto\boldsymbol{\omega}(\boldsymbol{v})=\omega
_{\boldsymbol{v}}:=\frac{1}{2}v^{\mathbf{c}}\omega_{\cdot\mathbf{c}\cdot
}^{\mathbf{a}\cdot\mathbf{b}}\theta_{\mathbf{a}}\wedge\theta_{\mathbf{b}}
\label{om2}%
\end{align}
and it is clear that
\begin{equation}
\omega(v)=\omega_{\boldsymbol{v}}. \label{om3}%
\end{equation}

One can immediately verify that for any $\mathcal{C}\in\sec\mathcal{C\ell
}(M,\mathtt{g})$ we have\footnote{For a rigorous derivation of this formula
using the concept of connections as 1-forms on a principal bundle with values
in a given Lie algebra see \cite{mr2004}.}
\begin{align}
\nabla_{\boldsymbol{v}}\mathcal{C}  &  =\eth_{\boldsymbol{v}}\mathcal{C}%
+\frac{1}{2}[\omega_{\boldsymbol{v}},\mathcal{C}]\label{code}\\
&  =\eth_{\boldsymbol{v}}\mathcal{C}+\omega_{\boldsymbol{v}}\times
\mathcal{C},\nonumber
\end{align}
where $\omega_{\boldsymbol{v}}\times\mathcal{C}:\mathcal{=}\frac{1}{2}%
(\omega_{\boldsymbol{v}}\mathcal{C-C}\omega_{\boldsymbol{v}})$ is the
commutator of sections of the Clifford bundle.

\begin{remark}
\label{vcommu}Note for future reference that if $v=\boldsymbol{g}%
(\boldsymbol{v}\mathbf{,}$ $)$ then
\begin{equation}
v\times\mathcal{C}=v\lrcorner\mathcal{C}. \label{code0}%
\end{equation}
Also take notice that
\begin{equation}
\nabla_{\boldsymbol{v}}\mathcal{C}=\boldsymbol{v}\cdot\overset{\triangledown
}{\mathbf{\partial}} \label{code1}%
\end{equation}

\end{remark}

\begin{definition}
The form curvature operator is the mapping\footnote{As well known the
curvature operator of a general connection $\nabla$ is the mapping
$\boldsymbol{\rho}:\sec\mathbf{(}TM\times TM)\rightarrow\mathrm{End}TM$,
$\boldsymbol{\rho}(\boldsymbol{u}\mathbf{,}\boldsymbol{v})=[\nabla
_{\boldsymbol{u}},\nabla_{\boldsymbol{v}}]-\nabla_{\lbrack\boldsymbol{u}%
\mathbf{,}\boldsymbol{v}]}$.}
\end{definition}

\begin{gather*}
\boldsymbol{\rho}:\sec\mathbf{(}\bigwedge\nolimits^{1}T^{\ast}M\times
\bigwedge\nolimits^{1}T^{\ast}M)\rightarrow\mathrm{End}\bigwedge
\nolimits^{1}T^{\ast}M,\\
\boldsymbol{\rho}(u,v)=[u\cdot\overset{\triangledown}{\mathbf{\partial}%
},v\cdot\overset{\triangledown}{\mathbf{\partial}}]-%
\bra
u,v%
\ket
\cdot\overset{\triangledown}{\mathbf{\partial}}\\
=[\nabla_{\boldsymbol{u}},\nabla_{\boldsymbol{v}}]-\nabla_{\lbrack
\boldsymbol{u}\mathbf{,}\boldsymbol{v}\mathbf{]}}%
\end{gather*}
with $u=\boldsymbol{g}(\boldsymbol{u}\mathbf{,}$ $)$, $v=\boldsymbol{g}%
(\boldsymbol{v}\mathbf{,}$ $)$, $\boldsymbol{u,v}\in\sec TU\subset\sec TM$

\begin{definition}
The form curvature extensor is the mapping
\end{definition}

\begin{gather*}
\mathfrak{\rho}:\sec\mathbf{(}\bigwedge\nolimits^{1}T^{\ast}M\times
\bigwedge\nolimits^{1}T^{\ast}M\times\bigwedge\nolimits^{1}T^{\ast
}M)\rightarrow\sec\bigwedge\nolimits^{1}T^{\ast}M,\\
\mathfrak{\rho}(u,v,w)=[u\cdot\overset{\triangledown}{\mathbf{\partial}%
},v\cdot\overset{\triangledown}{\mathbf{\partial}}]w-%
\bra
u,v%
\ket
\cdot\overset{\triangledown}{\mathbf{\partial}}w\\
=[\nabla_{\boldsymbol{u}},\nabla_{\boldsymbol{v}}]w-\nabla_{\lbrack
\boldsymbol{u}\mathbf{,}\boldsymbol{v}\mathbf{]}}w
\end{gather*}
with $u=\boldsymbol{g}(\boldsymbol{u}\mathbf{,}$ $)$, $v=\boldsymbol{g}%
(\boldsymbol{v}\mathbf{,}$ $)$, $w=\boldsymbol{g}(\boldsymbol{w}\mathbf{,}$
$)$, $\boldsymbol{u,v,w}\in\sec TU\subset TM$

It is obvious that for\ any Riemann-Cartan connection we have%
\begin{equation}
\mathfrak{\rho}(u,v,w)=-\mathfrak{\rho}(v,u,w), \label{pc1}%
\end{equation}

One can easily verify that for a Levi-Civita connection we have%
\begin{equation}
\mathfrak{\rho}(u,v,w)+\mathfrak{\rho}(v,w,u)+\mathfrak{\rho}(w,u,v)=0.
\label{pc2}%
\end{equation}
Note however that Eq.(\ref{pc2}) is not true for a general connection.

\begin{definition}
The mapping%
\begin{align}
\mathbf{R}  &  :\sec\mathbf{(}\bigwedge\nolimits^{1}T^{\ast}M)^{4}%
\rightarrow\sec\bigwedge\nolimits^{0}T^{\ast}M,\nonumber\\
\mathbf{R}(a,b,c,w)  &  =-\mathfrak{\rho}(a,b,c)\cdot w, \label{curv1}%
\end{align}
with $a=\boldsymbol{g}(\boldsymbol{a}\mathbf{,}$ $)$, $b=\boldsymbol{g}%
(\boldsymbol{b}\mathbf{,}$ $)$, $c=\boldsymbol{g}(\boldsymbol{c}\mathbf{,}$
$)$ $w=\boldsymbol{g}(\boldsymbol{w}\mathbf{,}$ $)$ and $\boldsymbol{u,v,w,c}%
\in\sec TU\subset TM$ is called the curvature tensor.
\end{definition}

One can show immediately that for the connection $\nabla$
\begin{equation}
\mathbf{R}(a,b,c,w)=-\mathbf{R}(b,a,c,w), \label{curv2}%
\end{equation}%
\begin{equation}
\mathbf{R}(a,b,c,w)=-\mathbf{R}(a,b,w,c), \label{curv3}%
\end{equation}
and that for a Levi-Civita connection
\begin{equation}
\mathbf{R}(a,b,c,w)=\mathbf{R}(c,w,a,b), \label{curv4}%
\end{equation}%
\begin{equation}
\mathbf{R}(a,b,c,w)+\mathbf{R}(b,c,a,w)+\mathbf{R}(c,a,b,w)=0, \label{curv5}%
\end{equation}
Equation \ref{curv5} is known as the first Bianchi identity.

\begin{proposition}
There exists a smooth $(2$-$2)$-extensor field,%
\begin{gather}
\mathfrak{R:}\sec%
{\textstyle\bigwedge\nolimits^{2}}
T^{\ast}M\rightarrow%
{\textstyle\bigwedge\nolimits^{2}}
T^{\ast}M,\nonumber\\
B\mapsto\mathfrak{R}(B) \label{curbiform}%
\end{gather}
called the curvature biform such that for any $a,v,c,d\in\sec%
{\textstyle\bigwedge\nolimits^{1}}
T^{\ast}M$ we have
\begin{equation}
\mathbf{R}(a,b,c,d)=\mathfrak{R}(a\wedge b)\cdot(c\wedge d)=-(c\wedge
d)\lrcorner\mathfrak{R}(a\wedge b) \label{curv6}%
\end{equation}
Such $B\mapsto\mathfrak{R}(B)$ is given by
\begin{equation}
\mathfrak{R}(B)=-\frac{1}{4}B\cdot(\partial_{a}\wedge\partial_{b})\partial
_{c}\wedge\partial_{d}\mathfrak{\rho}(a,b,c)\cdot d, \label{curv7}%
\end{equation}
and\ we also have%
\begin{equation}
\mathfrak{R}(a\wedge b)=-\frac{1}{2}\partial_{c}\wedge\partial_{d}%
\mathfrak{\rho}(a,b,c)\cdot d. \label{curv8}%
\end{equation}

\end{proposition}

\begin{proof}
First, we verify that Eq.(\ref{curv7}) and Eq.(\ref{curv8}) are indeed
equivalent. Indeed, Eq.(\ref{curv7}) implies Eq.(\ref{curv8}) since we have%
\begin{align*}
\mathfrak{R}(a\wedge b)  &  =-\frac{1}{4}(a\wedge b)\cdot(\partial_{p}%
\wedge\partial_{q})\partial_{\mathbf{c}}\wedge\partial_{\mathbf{d}%
}\mathfrak{\rho}(p,q,c)\cdot d\\
&  =-\frac{1}{4}\det\left[
\begin{array}
[c]{cc}%
a\cdot\partial_{p} & a\cdot\partial_{q}\\
b\cdot\partial_{p} & b\cdot\partial_{q}%
\end{array}
\right]  \partial_{\mathbf{c}}\wedge\partial_{\mathbf{d}}\mathfrak{\rho
}(p,q,c)\cdot d\\
&  =-\frac{1}{4}\left(  a\cdot\partial_{p}b\cdot\partial_{q}-a\cdot
\partial_{q}b\cdot\partial_{p}\right)  \partial_{\mathbf{c}}\wedge
\partial_{\mathbf{d}}\mathfrak{\rho}(p,q,c)\cdot d\\
&  =-\frac{1}{2}\left(  a\cdot\partial_{p}b\cdot\partial_{q}\right)
\partial_{\mathbf{c}}\wedge\partial_{\mathbf{d}}\mathfrak{\rho}(p,q,c)\cdot
d\\
&  =-\frac{1}{2}\partial_{c}\wedge\partial_{d}\mathfrak{\rho}(a,b,c)\cdot
d\text{ }.
\end{align*}
Also, Eq.(\ref{curv8}) implies Eq.(\ref{curv7}) since taking into account that%
\[
B=\frac{1}{2}B\cdot(\partial_{a}\wedge\partial_{b})a\wedge b
\]
we have%
\begin{align*}
\mathfrak{R}(B)  &  =\mathfrak{R(}\frac{1}{2}B\cdot(\partial_{a}\wedge
\partial_{b})a\wedge b)\\
&  =\frac{1}{2}B\cdot(\partial_{a}\wedge\partial_{b})\mathfrak{R}(a\wedge b)\\
&  =-\frac{1}{4}B\cdot(\partial_{a}\wedge\partial_{b})\partial_{c}%
\wedge\partial_{d}\mathfrak{\rho}(a,b,c)\cdot d\text{ }.
\end{align*}
Now, we show the validity of Eq.(\ref{curv6}). We have taking into account
Eq.(\ref{curv8})
\begin{align*}
\mathfrak{R}(a\wedge b)\lrcorner(c\wedge d)  &  =-\frac{1}{2}(c\wedge
d)\cdot(\partial_{p}\wedge\partial_{q})\mathfrak{\rho}(a,b,p)\cdot q\\
&  =-\frac{1}{2}\det\left[
\begin{array}
[c]{cc}%
c\cdot\partial_{p} & c\cdot\partial_{q}\\
d\cdot\partial_{p} & d\cdot\partial_{q}%
\end{array}
\right]  \mathfrak{\rho}(a,b,p)\cdot q\\
&  =-\frac{1}{2}\left(  c\cdot\partial_{p}d\cdot\partial_{q}-c\cdot
\partial_{q}d\cdot\partial_{p}\right)  \mathfrak{\rho}(a,b,p)\cdot q\\
&  =-c\cdot\partial_{p}d\cdot\partial_{q}\mathfrak{\rho}(a,b,p)\cdot q\\
&  =-\mathfrak{\rho}(a,b,c)\cdot d=\mathbf{R}(a,b,c,d),
\end{align*}
and the proposition is proved.
\end{proof}

\begin{proposition}
The curvature biform $\mathfrak{R}(u\wedge v)$ is given by\emph{\footnote{Note
that in Eq.(\ref{curvature}) $[\boldsymbol{u}\mathbf{,}\boldsymbol{v}]$ is the
standard Lie bracket of vector fields $\boldsymbol{u}$\textbf{ }and\textbf{
}$\boldsymbol{v}$ and $%
\bbra
u,v%
\bket
:=u\cdot\boldsymbol{\partial}v-v\cdot\boldsymbol{\partial}u$ is the commutator
of the 1-form fields $u$\textbf{ }and\textbf{ }$v$. More details if necessary
may be found in \cite{rodcap2007}.}}
\begin{equation}
\mathfrak{R}(u\wedge v)=u\cdot\eth\omega~(v)-v\cdot\eth\omega~(u)+\omega
(u)\times\omega(v). \label{rie1}%
\end{equation}

\end{proposition}

\begin{proof}
The proof is given in three steps \textbf{(a), (b) }and\textbf{ (c)}

(\textbf{a) }We first show that Eq.(\ref{rie1}) can be written as
\begin{align}
\mathfrak{R}(u\wedge v)  &  =u\cdot\overset{\triangledown}{\mathbf{\partial}%
}\boldsymbol{\omega}_{\boldsymbol{v}}-v\cdot\overset{\triangledown
}{\mathbf{\partial}}\boldsymbol{\omega}_{\boldsymbol{u}}-\frac{1}%
{2}[\boldsymbol{\omega}_{\boldsymbol{u}},\boldsymbol{\omega}_{\boldsymbol{v}%
}]-\boldsymbol{\omega}_{[\boldsymbol{u},\boldsymbol{v}]}\nonumber\\
&  =\nabla_{\boldsymbol{u}}\boldsymbol{\omega}_{\boldsymbol{v}}-\nabla
_{\boldsymbol{v}}\boldsymbol{\omega}_{\boldsymbol{u}}-\frac{1}{2}%
[\boldsymbol{\omega}_{\boldsymbol{u}},\boldsymbol{\omega}_{\boldsymbol{v}%
}]-\boldsymbol{\omega}_{[\boldsymbol{u},\boldsymbol{v}]}, \label{curvature}%
\end{align}
with $u=\boldsymbol{g}(\boldsymbol{u}\mathbf{,}$ $)$, $v=\boldsymbol{g}%
(\boldsymbol{v}\mathbf{,}$ $)$. Indeed,we have
\begin{equation}
u\cdot\overset{\triangledown}{\mathbf{\partial}}{(}\omega{(}v{))}=u\cdot
\eth{(}\omega{(}v{))+}\frac{1}{2}{[}\omega{(}u{),}\omega{(}v{)].}%
\end{equation}
and recalling the definition of the derivative of an extensor field, it is:%
\begin{equation}
\left(  u\cdot\eth\omega\right)  {{(}}v{{)}}\equiv u\cdot\eth\omega~{{{(}}%
}v{{{)}}}:=u{\cdot\eth{(}}\omega{{(}}v{{))}-}\omega{(}u\cdot\eth v{).}%
\end{equation}
we have,%
\begin{align}
u\cdot\eth{(}\omega{(}v{))-}v\cdot\eth{{(}}\omega{{(}}u{{))}}  &  =u\cdot
\eth\text{ }\omega{{(}}v{{)}-}v\cdot\eth\text{ }\omega{(}u{)+}\omega{(}%
u{\cdot}\eth v{)-}\omega{(}v{\cdot}\eth u{)}\nonumber\\
&  =u\cdot\eth\text{ }\omega{{(}}v{{)}-}v\cdot\eth\text{ }\omega{(}u{)+}%
\omega{(%
\bra
}u,v{%
\ket
)}\nonumber\\
&  =u\cdot\eth\text{ }\boldsymbol{\omega}{_{\boldsymbol{v}}}-v\cdot\eth\text{
}\boldsymbol{\omega}_{\boldsymbol{u}}+\boldsymbol{\omega}_{[\boldsymbol{u}%
,\boldsymbol{v}]},
\end{align}
and using the above equations in Eq.(\ref{rie1}) we arrive at
Eq.(\ref{curvature})

\textbf{(b)} Next we show (by finite induction) that for any $\mathcal{C}%
\in\sec\mathcal{C\ell}(M,\mathtt{g})$ we have
\begin{equation}
([\nabla_{\boldsymbol{u}},\nabla_{\boldsymbol{v}}]-\nabla_{\lbrack
\boldsymbol{u},\boldsymbol{v}]})\mathcal{C}=\frac{1}{2}[\mathfrak{R}(u\wedge
v),\mathcal{C}]\text{,} \label{exercise}%
\end{equation}
with $\mathfrak{R}(u\wedge v)$ given by Eq.(\ref{curvature}). Given that any
$\mathcal{C}\in\sec\mathcal{C\ell}(M,\mathtt{g})$ is a sum of nonhomogeneous
differential forms, i.e. $\mathcal{C}=%
{\textstyle\sum\nolimits_{p=0}^{n}}
\mathcal{C}_{p}$ with $\mathcal{C}_{p}\in\sec%
{\textstyle\bigwedge\nolimits^{r}}
T^{\ast}M\hookrightarrow\sec\mathcal{C\ell}(M,\mathtt{g})$ and taking into
account that $\mathcal{C}_{p}=\frac{1}{r!}\mathcal{C}_{\mathbf{i}_{1}%
\cdots\mathbf{i}_{p}}\theta^{\mathbf{i}_{1}}\cdots\theta^{\mathbf{i}_{p}}$ it
is enough to verify the formula for $p$-forms. We first verify the validity of
the formula for a $1$-form $\theta^{i}\in\sec\bigwedge\nolimits^{1}T^{\ast
}M\hookrightarrow\sec\mathcal{C\ell}(M,\mathtt{g})$. Using Eq.(\ref{curvature}%
) and the Jacobi identity%
\begin{equation}%
\bra
\boldsymbol{\omega}{_{\boldsymbol{v}}},%
\bra
\boldsymbol{\omega}{_{\boldsymbol{u}},}\theta^{\mathbf{i}}%
\ket
\ket
+%
\bra
\boldsymbol{\omega}{_{\boldsymbol{u}}},%
\bra
\theta^{\mathbf{i}}{,\boldsymbol{\omega}_{\boldsymbol{v}}}%
\ket
\ket
+%
\bra
\theta^{\mathbf{i}},%
\bra
\boldsymbol{\omega}_{\boldsymbol{v}},\boldsymbol{\omega}_{\boldsymbol{u}}%
\ket
\ket
=0,
\end{equation}
we have that%
\begin{align}
&  \frac{1}{2}[\mathfrak{R}(u\wedge v),\theta^{\mathbf{i}}]\label{caset}\\
&  =\frac{1}{2}\left\{  (\nabla_{\boldsymbol{u}}\omega{_{\boldsymbol{v}}%
}\theta^{\mathbf{i}}-\theta^{\mathbf{i}}\nabla_{\boldsymbol{u}}\omega
{_{\boldsymbol{v}}+}\frac{1}{2}{[[}\boldsymbol{\omega}{{_{\boldsymbol{v}}%
},\boldsymbol{\omega}{_{\boldsymbol{u}}}],\theta^{\mathbf{i}}]-(}%
\nabla_{\boldsymbol{v}}\boldsymbol{\omega}{_{\boldsymbol{u}}}\theta
^{\mathbf{i}}+\theta^{\mathbf{i}}\nabla_{\boldsymbol{v}}\boldsymbol{\omega
}{_{\boldsymbol{u}}-[\boldsymbol{\omega}{_{[\boldsymbol{u}\mathbf{,}%
\boldsymbol{v}]},}\theta^{\mathbf{i}}]}\right\} \nonumber\\
&  =\frac{1}{2}\left\{  \left[  \nabla_{\boldsymbol{u}}\boldsymbol{\omega
}{_{\boldsymbol{v}}},\theta^{\mathbf{i}}\right]  +\frac{1}{2}%
[\boldsymbol{\omega}{_{\boldsymbol{v}},[\boldsymbol{\omega}_{\boldsymbol{u}%
},\theta^{\mathbf{i}}]]-}\left[  \nabla_{\boldsymbol{v}}\boldsymbol{\omega
}{_{\boldsymbol{v}}},\theta^{\mathbf{i}}\right]  -\frac{1}{2}%
[\boldsymbol{\omega}{_{\boldsymbol{u}},[\boldsymbol{\omega}_{\boldsymbol{v}%
},\theta^{\mathbf{i}}]]-[\boldsymbol{\omega}_{[\boldsymbol{u},\boldsymbol{v}%
]},\theta^{\mathbf{i}}]}\right\} \nonumber\\
&  =\nabla_{\boldsymbol{u}}{{\nabla}_{\boldsymbol{v}}}\theta^{\mathbf{i}%
}-\nabla_{\boldsymbol{v}}{{\nabla}_{\boldsymbol{u}}}\theta^{\mathbf{i}}%
-\nabla_{\lbrack\boldsymbol{u}\mathbf{,}\boldsymbol{v}]}\theta^{\mathbf{i}}.
\end{align}
Now, suppose the formula is valid for $p$-forms. Let us calculate the first
member of Eq.(\ref{exercise}) for the $(r+1)$-form $\theta^{\mathbf{i}%
_{1}\cdots\mathbf{i}_{r+1}}=\theta^{\mathbf{i}_{1}}\theta^{\mathbf{i}_{2}%
}\cdots\theta^{\mathbf{i}_{r+1}}$. We have .
\begin{align}
&  \nabla_{\boldsymbol{u}}{{\nabla}_{\boldsymbol{v}}}(\theta^{\mathbf{i}%
_{1}\cdots\mathbf{i}_{r+1}})-\nabla_{\boldsymbol{v}}{{\nabla}_{\boldsymbol{u}%
}}(\theta^{\mathbf{i}_{1}\cdots\mathbf{i}_{r+1}})-\nabla_{\lbrack
\boldsymbol{u}\mathbf{,}\boldsymbol{v}]}(\theta^{\mathbf{i}_{1}\cdots
\mathbf{i}_{r+1}})\nonumber\\
&  =\nabla_{\boldsymbol{u}}((\nabla_{\boldsymbol{v}}{\theta^{\mathbf{i}_{1}}%
)}\theta^{\mathbf{i}_{2}{\cdots}\mathbf{i}_{r+1}}+{\theta^{\mathbf{i}_{1}%
}\nabla_{\boldsymbol{v}}}\theta^{\mathbf{i}_{2}{\cdots}\mathbf{i}_{r+1}%
})-\nabla_{\boldsymbol{v}}((\nabla_{\boldsymbol{u}}{\theta^{\mathbf{i}_{1}}%
)}\theta^{\mathbf{i}_{2}{\cdots}\mathbf{i}_{r+1}}+{\theta^{\mathbf{i}_{1}%
}\nabla_{\boldsymbol{u}}}\theta^{\mathbf{i}_{2}{\cdots}\mathbf{i}_{r+1}%
})\nonumber\\
&  -(\nabla_{\lbrack\boldsymbol{u}\mathbf{,}\boldsymbol{v}]}\theta
^{\mathbf{i}_{1}})\theta^{\mathbf{i}_{2}{\cdots}\mathbf{i}_{r+1}}%
)-\theta^{\mathbf{i}_{1}}\nabla_{\lbrack\boldsymbol{u}\mathbf{,}%
\boldsymbol{v}]}\theta^{\mathbf{i}_{2}{\cdots}\mathbf{i}_{r+1}}\nonumber\\
&  =(\nabla_{\boldsymbol{u}}{{\nabla}_{\boldsymbol{v}}\theta^{\mathbf{i}_{1}%
})}\theta^{\mathbf{i}_{2}{\cdots}\mathbf{i}_{r+1}}+{{\nabla}_{\boldsymbol{v}%
}\theta^{\mathbf{i}_{1}}}\nabla_{\boldsymbol{u}}\theta^{\mathbf{i}_{2}{\cdots
}\mathbf{i}_{r+1}}+\nabla_{\boldsymbol{u}}{\theta^{\mathbf{i}_{1}}{\nabla
}_{\boldsymbol{v}}}\theta^{\mathbf{i}_{2}{\cdots}\mathbf{i}_{r+1}}%
+{\theta^{\mathbf{i}_{1}}\nabla_{\boldsymbol{u}}{\nabla}_{\boldsymbol{v}}%
}\theta^{\mathbf{i}_{2}{\cdots}\mathbf{i}_{r+1}}\nonumber\\
&  -(\nabla_{\boldsymbol{v}}{{\nabla}_{\boldsymbol{u}}\theta^{\mathbf{i}_{1}%
})}\theta^{\mathbf{i}_{2}{\cdots}\mathbf{i}_{r+1}}-{{\nabla}_{\boldsymbol{u}%
}\theta^{\mathbf{i}_{1}}}\nabla_{\boldsymbol{v}}\theta^{\mathbf{i}_{2}{\cdots
}\mathbf{i}_{r+1}}-\nabla_{\boldsymbol{v}}{\theta^{\mathbf{i}_{1}}{\nabla
}_{\boldsymbol{u}}}\theta^{\mathbf{i}_{2}{\cdots}\mathbf{i}_{r+1}}%
-{\theta^{\mathbf{i}_{1}}}\nabla_{\boldsymbol{v}}{{\nabla}_{\boldsymbol{u}}%
}\theta^{\mathbf{i}_{2}{\cdots}\mathbf{i}_{r+1}}\nonumber\\
&  -({\nabla_{\lbrack\boldsymbol{u}\mathbf{,}\boldsymbol{v}]}}\theta
^{\mathbf{i}_{1}})\theta^{\mathbf{i}_{2}{\cdots}\mathbf{i}_{r+1}}%
)-\theta^{\mathbf{i}_{1}}{\nabla_{\lbrack\boldsymbol{u}\mathbf{,}%
\boldsymbol{v}]}}\theta^{\mathbf{i}_{2}{\cdots}\mathbf{i}_{r+1}}\nonumber\\
&  ={\theta^{\mathbf{i}_{1}}(\nabla_{\boldsymbol{u}}{{\nabla}_{\boldsymbol{v}%
}}\theta^{\mathbf{i}_{2}{\cdots}\mathbf{i}_{r+1}}-\nabla_{\boldsymbol{v}%
}{{\nabla}_{\boldsymbol{u}}}\theta^{\mathbf{i}_{2}{\cdots}\mathbf{i}_{r+1}%
}-\nabla_{\lbrack\boldsymbol{u}\mathbf{,}\boldsymbol{v}]}\theta^{\mathbf{i}%
_{2}{\cdots}\mathbf{i}_{r+1}})}\nonumber\\
&  {+(\nabla_{\boldsymbol{u}}{{\nabla}_{\boldsymbol{v}}\theta^{\mathbf{i}_{1}%
}-}\nabla_{\boldsymbol{v}}{{\nabla}_{\boldsymbol{u}}\theta^{\mathbf{i}_{1}}%
-}\nabla_{\lbrack\boldsymbol{u}\mathbf{,}\boldsymbol{v}]}\theta^{\mathbf{i}%
_{1}})}\theta^{\mathbf{i}_{2}{\cdots}\mathbf{i}_{r+1}}\nonumber\\
&  ={\theta^{\mathbf{i}_{1}}(\frac{1}{2}[\mathfrak{R}(u\wedge v),\theta
^{\mathbf{i}_{2}{\cdots}\mathbf{i}_{r+1}}])+(\frac{1}{2}[\mathfrak{R}(u\wedge
v),\theta^{\mathbf{i}_{1}}])}\theta^{\mathbf{i}_{2}{\cdots}\mathbf{i}_{r+1}%
}\nonumber\\
&  ={\frac{1}{2}[\mathfrak{R}(u\wedge v),\theta^{\mathbf{i}_{1}{\cdots
}\mathbf{i}_{r+1}}],} \label{pform}%
\end{align}
where the last line of Eq.(\ref{pform}) is the second member of
Eq.(\ref{exercise}) evaluated for ${\theta^{\mathbf{i}_{1}}\theta
^{\mathbf{i}_{2}}\cdots\theta^{\mathbf{i}_{r+1}}}$.

(\textbf{c}) Now, it remains to verify that
\begin{equation}
\mathbf{R(}u,v,t,z\mathbf{)}=\mathbf{(}t\wedge z)\cdot\mathfrak{R}(u\wedge v),
\label{rie2}%
\end{equation}
with $\mathfrak{R}(u\wedge v)$ given by Eq.(\ref{curvature}). Indeed, from a
well known identity, we have that for any $t,z\in\sec%
{\textstyle\bigwedge\nolimits^{1}}
T^{\ast}M$, and $\mathfrak{R}(u\wedge v)\in\sec%
{\textstyle\bigwedge\nolimits^{2}}
T^{\ast}M$\ it is
\begin{align*}
(z\wedge t)\cdot\mathfrak{R}(u\wedge v)  &  =\mathbf{-}z\lrcorner
(t\lrcorner\mathfrak{R}(u\wedge v))\\
&  =z\cdot(\mathfrak{R}(u\wedge v)\llcorner t)\\
&  =\frac{1}{2}z\cdot\lbrack\mathfrak{R}(u,v),t]\\
&  \overset{\text{Eq.(\ref{exercise})}}{=}z\cdot(\nabla_{\boldsymbol{u}%
}{{\nabla}_{\boldsymbol{v}}}t-\nabla_{\boldsymbol{v}}{{\nabla}_{\boldsymbol{u}%
}}t-\nabla_{\lbrack\boldsymbol{u}\mathbf{,}\boldsymbol{v}]}t)
\end{align*}
and the proposition is proved.
\end{proof}

In particular we have:%
\begin{align}
\mathbf{R(}u,v,z,t\mathbf{)}  &  =z_{\mathbf{c}}t^{\mathbf{d}}u^{\mathbf{a}%
}v^{\mathbf{b}}R_{\cdot\mathbf{dab}}^{\mathbf{c\cdot\cdot\cdot}},\nonumber\\
R_{\cdot\mathbf{cab}}^{\mathbf{d\cdot\cdot\cdot}}  &  =\boldsymbol{e}%
_{\mathbf{a}}(\omega_{\mathbf{\cdot bc}}^{\mathbf{d\cdot\cdot}}%
)-\boldsymbol{e}_{\mathbf{b}}(\omega_{\mathbf{\cdot ac}}^{\mathbf{d\cdot\cdot
}})+\omega_{\mathbf{\cdot ak}}^{\mathbf{d\cdot\cdot}}\omega_{\mathbf{\cdot
bc}}^{\mathbf{k\cdot\cdot}}-\omega_{\mathbf{\cdot bk}}^{\mathbf{d\cdot\cdot}%
}\omega_{\mathbf{\cdot ac}}^{\mathbf{k\cdot\cdot}}-c_{\mathbf{\cdot ab}%
}^{\mathbf{k\cdot\cdot}}\omega_{\mathbf{\cdot kc}}^{\mathbf{d\cdot\cdot}%
}\text{ }. \label{riemann impo}%
\end{align}

and
\begin{equation}
\mathbf{R(}\theta^{\mathbf{a}},\theta^{\mathbf{b}},\theta_{\mathbf{a}}%
,\theta_{\mathbf{b}}\mathbf{)=(}\theta^{\mathbf{a}}\wedge\theta^{\mathbf{b}%
})\cdot\mathfrak{R}(\theta_{\mathbf{a}}\wedge\theta_{\mathbf{b}})=R,
\label{curv scalar}%
\end{equation}
where $R$ is the curvature scalar.

\begin{proposition}
For any $v\in\sec\bigwedge\nolimits^{1}T^{\ast}M\hookrightarrow\sec
\mathcal{C}\ell(M,\mathtt{g})$
\begin{equation}
\lbrack\nabla_{\boldsymbol{e}_{\mathbf{a}}},\nabla_{\boldsymbol{e}%
_{\mathbf{b}}}]v=\mathfrak{R}(\theta_{\mathbf{a}}\wedge\theta_{\mathbf{b}%
})\llcorner v-(T_{\mathbf{ab}}^{\mathbf{c}}-\omega_{\mathbf{\cdot ab}%
}^{\mathbf{c\cdot\cdot}}+\omega_{\mathbf{\cdot ba}}^{\mathbf{c\cdot\cdot}%
})\nabla_{\boldsymbol{e}_{\mathbf{c}}}v. \label{commut ident 1}%
\end{equation}

\end{proposition}

\begin{proof}
From Eq.(\ref{exercise}) we can write%
\begin{gather*}
\lbrack\nabla_{\boldsymbol{e}_{\mathbf{a}}},\nabla_{\boldsymbol{e}%
_{\mathbf{b}}}]v=\frac{1}{2}[\mathfrak{R}(\theta_{\mathbf{a}}\wedge
\theta_{\mathbf{b}}),v]+\nabla_{\lbrack\boldsymbol{e}_{\mathbf{a}}%
\mathbf{,}\boldsymbol{e}_{\mathbf{b}}]}v\\
=\mathfrak{R}(\theta_{\mathbf{a}}\wedge\theta_{\mathbf{b}})\llcorner
v+\nabla_{([\boldsymbol{e}_{\mathbf{a}}\mathbf{,}\boldsymbol{e}_{\mathbf{b}%
}]-\nabla_{\boldsymbol{e}_{a}\boldsymbol{e}\mathbf{b}}+\nabla_{\boldsymbol{e}%
_{\mathbf{b}}}\boldsymbol{e}_{\mathbf{a}})}v+\nabla_{\nabla_{\boldsymbol{e}%
_{\mathbf{a}}}\boldsymbol{e}_{\mathbf{b}}}v-\nabla_{\nabla_{\boldsymbol{e}%
_{\mathbf{b}}}\boldsymbol{e}_{\mathbf{a}}}v\\
=\mathfrak{R}(\theta_{\mathbf{a}}\wedge\theta_{\mathbf{b}})\llcorner
v+\nabla_{-T_{\cdot\mathbf{ab}}^{\mathbf{c}\cdot\cdot}\boldsymbol{e}%
_{\mathbf{c}}}v+\nabla_{\omega_{\cdot\mathbf{ab}}^{\mathbf{c}\cdot\cdot
}\boldsymbol{e}_{\mathbf{c}}}v-\nabla_{\omega_{\cdot\mathbf{ba}}%
^{\mathbf{c}\cdot\cdot}\boldsymbol{e}_{\mathbf{c}}}v\\
=\mathfrak{R}(\theta_{\mathbf{a}}\wedge\theta_{\mathbf{b}})\llcorner
v-(T_{\cdot\mathbf{ab}}^{\mathbf{c}\cdot\cdot}-\omega_{\cdot\mathbf{ab}%
}^{\mathbf{c}\cdot\cdot}+\omega_{\cdot\mathbf{ba}}^{\mathbf{c}\cdot\cdot
})\nabla_{\boldsymbol{e}_{\mathbf{c}}}v
\end{gather*}
which proves the proposition.
\end{proof}

\begin{proposition}%
\begin{equation}
\mathfrak{R}(\theta^{\mathbf{a}}\wedge\theta_{\mathbf{b}})=\mathcal{R}%
_{\mathbf{\cdot b}}^{\mathbf{a\cdot}}=d\omega_{\cdot\mathbf{b}}%
^{\mathbf{a\cdot}}+\omega_{\cdot\mathbf{c}}^{\mathbf{a\cdot}}\wedge
\omega_{\cdot\mathbf{b}}^{\mathbf{c\cdot}} \label{rR}%
\end{equation}

\end{proposition}

\begin{proof}
Recall that using
\begin{align}
([\nabla_{\boldsymbol{e}_{\mathbf{k}}},\nabla_{\boldsymbol{e}_{\mathbf{l}}%
}]-\nabla_{\lbrack\boldsymbol{e}_{\mathbf{k}}\mathbf{,}\boldsymbol{e}%
_{\mathbf{l}}]})\theta^{\mathbf{j}}  &  =\mathbf{\rho(}\boldsymbol{e}%
\mathbf{_{\mathbf{k}}\mathbf{,}}\boldsymbol{e}\mathbf{_{\mathbf{l}})}%
\theta^{\mathbf{j}}=-R_{\cdot\mathbf{ikl}}^{\mathbf{j\cdot\cdot\cdot}}%
\theta^{\mathbf{i}},\nonumber\\
([\nabla_{\boldsymbol{e}_{\mathbf{k}}},\nabla_{\boldsymbol{e}_{\mathbf{l}}%
}]-\nabla_{\lbrack\boldsymbol{e}_{\mathbf{k}}\mathbf{,}\boldsymbol{e}%
_{\mathbf{l}}]})\theta_{\mathbf{j}}  &  =\mathbf{\rho(}\boldsymbol{e}%
\mathbf{_{\mathbf{k}}\mathbf{,}}\boldsymbol{e}\mathbf{_{\mathbf{l}})}%
\theta_{\mathbf{j}}=R_{\cdot\mathbf{jkl}}^{\mathbf{i\cdot\cdot\cdot}}%
\theta_{\mathbf{i}}, \label{riem}%
\end{align}
we have%
\begin{equation}
\mathfrak{R}(\theta_{\mathbf{a}}\wedge\theta_{\mathbf{b}})\llcorner
v=v^{\mathbf{m}}\mathbf{\rho}(\boldsymbol{e}_{\mathbf{a}},\boldsymbol{e}%
_{\mathbf{b}})\theta_{\mathbf{m}}=v^{\mathbf{m}}R_{\cdot\mathbf{mab}%
}^{\mathbf{i\cdot\cdot\cdot}}\theta_{\mathbf{i}}. \label{rR2}%
\end{equation}

On the other hand, for a general connection, we must write
\begin{equation}
\mathcal{R}_{\mathbf{ab}}:=\frac{1}{2}R_{\mathbf{klab}}\theta^{\mathbf{k}%
}\wedge\theta^{\mathbf{l}} \label{2-formg}%
\end{equation}
and then
\[
\mathcal{R}_{\mathbf{ab}}\llcorner v=\frac{1}{2}v^{\mathbf{m}}R_{\mathbf{klab}%
}(\theta^{\mathbf{k}}\wedge\theta^{\mathbf{l}})\llcorner\theta_{\mathbf{m}%
}=-v^{\mathbf{m}}R_{\mathbf{mlab}}\theta^{\mathbf{l}}=v^{\mathbf{m}%
}R_{\mathbf{lmab}}\theta^{\mathbf{l}}=v^{\mathbf{m}}R_{\cdot\mathbf{mab}%
}^{\mathbf{l\cdot\cdot\cdot}}\theta_{\mathbf{l}}%
\]
and the proposition is proved.
\end{proof}

\begin{proposition}
\label{RICCI}The Ricci 1-forms\footnote{The $R_{\mathbf{b}}^{\mathbf{a}}$
$:=\eta^{\mathbf{ca}}R_{\cdot\mathbf{ckb}}^{\mathbf{k\cdot\cdot\cdot}}$are the
components of the Ricci tensor.} $\mathcal{R}^{\mathbf{d}}:=R_{\mathbf{b}%
}^{\mathbf{d}}\theta^{\mathbf{b}}$ and the curvature biform $\mathfrak{R}%
(\theta_{\mathbf{a}}\wedge\theta_{\mathbf{b}})$ for the Levi-Civita connection
$D$\ of $\boldsymbol{g}$ are related by,%
\begin{equation}
\mathcal{R}^{\mathbf{d}}=\frac{1}{2}(\theta^{\mathbf{a}}\wedge\theta
^{\mathbf{b}})(\mathfrak{R}(\theta_{\mathbf{a}}\wedge\theta_{\mathbf{b}%
})\llcorner\theta^{\mathbf{d}}) \label{riccixcurv}%
\end{equation}

\end{proposition}

\begin{proof}
Recalling that the Ricci operator is given by\footnote{See Chapter 4 of
\cite{rodcap2007}.}%
\begin{equation}
\boldsymbol{\partial}{\wedge}\boldsymbol{\partial}\theta^{\mathbf{d}}=\frac
{1}{2}(\theta^{\mathbf{a}}\wedge\theta^{\mathbf{b}})\left(  [D_{\boldsymbol{e}%
_{\mathbf{a}}},D_{\boldsymbol{e}_{\mathbf{b}}}]\theta^{\mathbf{d}%
}-c_{\mathbf{\cdot ab}}^{\mathbf{c\cdot\cdot}}D_{\boldsymbol{e}_{\mathbf{c}}%
}\theta^{\mathbf{d}}\right)  \label{ri1}%
\end{equation}
and moreover taking into account that by the first Bianchi identity it is
$\mathcal{R}_{\cdot d}^{\mathbf{c}\cdot}\wedge\theta_{\mathbf{c}}=0$, we have
\begin{align}
\frac{1}{2}(\theta^{\mathbf{a}}\wedge\theta^{\mathbf{b}})\left(
[D_{\boldsymbol{e}_{\mathbf{a}}},D_{\boldsymbol{e}_{\mathbf{b}}}%
]\theta^{\mathbf{d}}-c_{\cdot\mathbf{ab}}^{\mathbf{c\cdot}}D_{\boldsymbol{e}%
_{\mathbf{c}}}\theta^{\mathbf{d}}\right)   &  =-\frac{1}{2}(\theta
^{\mathbf{a}}\wedge\theta^{\mathbf{b}})R_{\cdot\mathbf{cab}}^{\mathbf{d\cdot
\cdot\cdot}}\theta^{\mathbf{c}}=\mathcal{R}^{\mathbf{cd}}\theta_{\mathbf{c}%
}\nonumber\\
&  =\mathcal{R}^{\mathbf{cd}}\lrcorner\theta_{\mathbf{c}}+\mathcal{R}%
^{\mathbf{cd}}\wedge\theta_{\mathbf{c}}=-\theta_{\mathbf{c}}\lrcorner
\mathcal{R}^{\mathbf{cd}}\nonumber\\
&  =-\frac{1}{2}\theta_{\mathbf{c}}\lrcorner(\theta^{\mathbf{a}}\wedge
\theta^{\mathbf{b}})R_{\cdot\cdot\mathbf{ab}}^{\mathbf{cd\cdot\cdot}%
}\nonumber\\
&  =-R_{\cdot\cdot\mathbf{cb}}^{\mathbf{cd\cdot\cdot}}\theta^{\mathbf{b}%
}=\mathcal{R}^{\mathbf{d}} \label{ri2}%
\end{align}
which proves the proposition.
\end{proof}

Proposition \ref{RICCI} suggests the

\begin{definition}
The Ricci extensor is the mapping
\begin{gather}
\mathcal{R}:\sec\bigwedge\nolimits^{1}T^{\ast}M\rightarrow\sec\bigwedge
\nolimits^{1}T^{\ast}M,\nonumber\\
\mathcal{R}(v)=\partial_{u}\mathcal{R}(u\wedge v). \label{ricciext}%
\end{gather}

\end{definition}

\begin{remark}
Of course, we must have $\mathcal{R}(\theta^{\mathbf{d}})=\mathcal{R}%
^{\mathbf{d}}$. Moreover, we have%
\begin{align*}
\partial_{u}\mathfrak{R}(u\wedge v)  &  =\theta^{\mathbf{b}}\frac{\partial
}{\partial u^{\mathbf{b}}}\mathfrak{R}(u_{\mathbf{k}}\theta^{\mathbf{k}}\wedge
v)=\theta^{\mathbf{b}}\frac{\partial}{\partial u^{\mathbf{b}}}u^{\mathbf{k}%
}\mathfrak{R}(\theta_{\mathbf{k}}\wedge v)\\
&  =\theta^{\mathbf{b}}\mathfrak{R}(\delta_{\mathbf{b}}^{\mathbf{k}}%
\theta_{\mathbf{k}}\wedge v)=\theta^{\mathbf{b}}\mathfrak{R}(\theta
_{\mathbf{b}}\wedge v)\\
&  =\theta^{\mathbf{b}}\lrcorner\mathfrak{R}(\theta_{\mathbf{b}}\wedge
v)+\theta^{\mathbf{b}}\wedge\mathfrak{R}(\theta_{\mathbf{b}}\wedge v)\\
&  =\theta^{\mathbf{b}}\lrcorner\mathfrak{R}(\theta_{\mathbf{b}}\wedge v).
\end{align*}
So,
\begin{equation}
\partial_{u}\mathfrak{R}(u\wedge v)=\partial_{u}\lrcorner\mathfrak{R}(u\wedge
v)\text{ and }\partial_{u}\wedge\mathfrak{R}(u\wedge v)=0. \label{proctness}%
\end{equation}

\end{remark}

\section{The Riemannian or Semi-Riemannian Geometry of a Submanifold $M$ of
$\mathring{M}$}

\subsection{Motivation}

Any manifold $M,\dim M=m$, according to Whitney's theorem (see, e.g.,
\cite{am}), can be realized as a submanifold of $\mathbb{R}^{n}$, with $n=2m$.
However, if $M$ carries additional structure the number $n$ in general must be
greater than $2m$. Indeed, it has been shown by Eddington \cite{eddington}
that if dim $M=4$ and if $M$ carries a Lorentzian metric $\boldsymbol{g}$ and
which moreover satisfies Einstein's equations, then $M$ can be \emph{locally}
embedded in a (pseudo)Euclidean space $\mathbb{R}^{1,9}$. Also, isometric
embeddings of general Lorentzian spacetimes would require a lot of extra
dimensions \cite{clarke}. Indeed, a compact Lorentzian manifold can be
embedded isometrically in $\mathbb{R}^{2,46}$ and a non-compact one can be
embedded isometrically in $\mathbb{R}^{2,87}$! In particular this last result
shows that the spacetime of M-theory \cite{duff,horava} may not be large
enough to contain 4-dimensional branes with arbitrary metric tensors. In what
follows we show how to relate the intrinsic differential geometry of a
structure $(M,\boldsymbol{g},D)$ where $\boldsymbol{g}$ is a metric of
signature $(p,q)$, $D\ $is its Levi-Civita connection and $M$ is an orientable
\emph{proper } submanifold of $\mathring{M}$, i.e., there is defined on $M$ a
global volume element $\tau_{\boldsymbol{g}}=I_{m}$ whose expression on
$U\subset M$ is given by
\begin{equation}
I_{m}=\theta^{1}\theta^{2}\cdots\theta^{m}. \label{VOLUME}%
\end{equation}
We suppose moreover that $\mathring{M}\simeq\mathbb{R}^{n}$ and it is equipped
with a metric $\boldsymbol{\mathring{g}}$ of signature $(\mathring
{p},\mathring{q})=n$. However, take notice that our presentation in the form
of a local theory is easily adapted for a general manifold $\mathring{M}$.

\subsubsection{Projection Operator $\mathbf{P}$}

\begin{definition}
\label{projope}Let $\mathcal{C}=%
{\textstyle\sum\nolimits_{r=0}^{n}}
\mathcal{C}_{r}$, with $\mathcal{C}_{r}\in\sec%
{\textstyle\bigwedge\nolimits^{r}}
T^{\ast}\mathring{M}\hookrightarrow\sec\mathcal{C}\ell(\mathring
{M},\mathtt{\mathring{g}})$. The Projection operator on $M$ is the extensor
field%
\begin{gather}
\mathbf{P}:\sec\mathcal{C}\ell(\mathring{M},\mathtt{\mathring{g}}%
)\rightarrow\sec\mathcal{C}\ell(M,\mathtt{g}),\nonumber\\
\mathbf{P}(\mathcal{C})=(\mathcal{C}\lrcorner I_{m})I_{m}^{-1}. \label{proj}%
\end{gather}

\end{definition}

\begin{remark}
\label{outside}Note that$\ $for all $\mathcal{C}_{k}\in\sec%
{\textstyle\bigwedge\nolimits^{k}}
T^{\ast}\mathring{M}\hookrightarrow\sec\mathcal{C}\ell(\mathring
{M},\mathtt{\mathring{g}})$, if $k>m$ then $\mathbf{P}(\mathcal{C}_{k})=0$,
but of course, it may happen that even if $A_{r}\in\sec%
{\textstyle\bigwedge\nolimits^{r}}
T^{\ast}\mathring{M}\hookrightarrow\sec\mathcal{C}\ell(M,\mathtt{\mathring{g}%
})$ with $r\leq m$ we \ may have $\mathbf{P}(A_{r})=0.$
\end{remark}

We define the complement of $\mathbf{P}$ by
\begin{equation}
\mathbf{P}_{\perp}(\mathcal{C})=\mathcal{C}-\mathbf{P}(\mathcal{C})
\label{compl}%
\end{equation}
and it is clear that $\mathbf{P}_{\perp}(\mathcal{C})$ have only components
lying outside $\mathcal{C}\ell(M,\mathtt{g})$.\ It is quite clear also that
any $\mathcal{C}$ with components not all belonging to $\sec\mathcal{C}%
\ell(\mathring{M},\mathtt{\mathring{g}})$ will satisfy $\mathcal{C}\lrcorner
I_{m}=0.$

Having introduced in Section 1 the derivative operators $\boldsymbol{\mathring
{\partial}}$ and its restriction $\mathfrak{\mathring{d}}=\left.
\boldsymbol{\mathring{\partial}}\right\vert _{M}$ (Eq. (\ref{1}) and
Eq.(\ref{2a})) we extend the action of $\mathbf{P}$ to act on the operator
$\mathfrak{\mathring{d}}$, defining:
\begin{equation}
\mathbf{P}(\mathfrak{\mathring{d}})=%
{\textstyle\sum\nolimits_{\mathbf{k}=1}^{m}}
\mathbf{P}(\theta^{\mathbf{k}}\mathring{D}_{\boldsymbol{e}_{\mathbf{k}}}):=%
{\textstyle\sum\nolimits_{\mathbf{k}=1}^{m}}
\mathbf{P}(\theta^{\mathbf{k}})\mathring{D}_{\boldsymbol{e}_{\mathbf{k}}}=%
{\textstyle\sum\nolimits_{\mathbf{k}=1}^{m}}
\theta^{\mathbf{k}}\mathring{D}_{\boldsymbol{e}_{\mathbf{k}}}%
=\mathfrak{\mathring{d}}. \label{p1}%
\end{equation}

\subsubsection{Shape Operator $\mathbf{S}$}

\begin{definition}
Given $\mathcal{C}\in\sec\mathcal{C}\ell(\mathring{M},\mathtt{\mathring{g}}%
)$\ we define the shape operator
\begin{gather}
\mathbf{S:}\sec\mathcal{C}\ell(\mathring{M},\mathtt{\mathring{g}}%
)\rightarrow\sec\mathcal{C}\ell(\mathring{M},\mathtt{\mathring{g}%
}),\nonumber\\
\mathbf{S}(\mathcal{C})=\mathfrak{\mathring{d}}\mathbf{P}\text{~}%
(\mathcal{C})=\mathfrak{\mathring{d}(}\mathbf{P}(\mathcal{C}))-\mathbf{P}%
\mathfrak{(\mathring{d}}\mathcal{C}). \label{S1}%
\end{gather}

\end{definition}

For any $\mathcal{C\in}\sec\mathcal{C}\ell(M,\mathtt{g})$ and
$\boldsymbol{v\in}\sec TM$ we write as usual \cite{choquet,hicks}
\begin{equation}
\mathring{D}_{\boldsymbol{v}}\mathcal{C}=(\mathring{D}_{\boldsymbol{v}%
}\mathcal{C)}_{\parallel}+(\mathring{D}_{\boldsymbol{v}}\mathcal{C)}_{\perp}
\label{SS1}%
\end{equation}
where $(\mathring{D}_{\boldsymbol{v}}\mathcal{C)}_{\parallel}\in
\sec\mathcal{C}\ell(M,\mathtt{g})$ and $(\mathring{D}_{\boldsymbol{v}%
}\mathcal{C)}_{\perp}\in\sec[\mathcal{C}\ell(M,\mathtt{g})]_{\perp}$ where
$[\mathcal{C}\ell(M,\mathtt{g})]_{\perp}$ is the orthogonal complement of
$\mathcal{C}\ell(\mathring{M},\mathtt{\mathring{g}})$ in $\mathcal{C}%
\ell(M,\mathtt{g})$.

As it is very well known \cite{choquet,hicks} \ if $\boldsymbol{g}%
:=\boldsymbol{i}^{\ast}\boldsymbol{\mathring{g}}$ and $\boldsymbol{v}\in\sec
TM$ (and $v=\boldsymbol{g}(\boldsymbol{v},)$) and $\mathcal{C}\in
\sec\mathcal{C}\ell(\mathring{M},\mathtt{\mathring{g}})$ the Levi-Civita
connection $D$ of $\boldsymbol{g}:=\boldsymbol{i}^{\ast}\boldsymbol{\mathring
{g}}$ is given by
\begin{equation}
D_{\boldsymbol{v}}\mathcal{C}:=(\mathring{D}_{\boldsymbol{v}}\mathcal{C)}%
_{\parallel} \label{p2}%
\end{equation}
and of course%
\begin{equation}
D_{\boldsymbol{v}}\mathcal{C}:=(v\cdot\boldsymbol{\mathring{\partial}%
}\mathcal{C})_{\parallel} \label{ss3}%
\end{equation}

Moreover, note that we can write for any $\mathcal{C}\in\sec\mathcal{C}%
\ell(M,\mathtt{g})$
\begin{equation}
v\cdot\boldsymbol{\partial}\mathcal{C}=(v\cdot\boldsymbol{\mathring{\partial}%
}\mathcal{C})_{\parallel}=\mathbf{P}(v\cdot\boldsymbol{\mathring{\partial}%
}\mathcal{C}) \label{p3}%
\end{equation}

Also, writing
\begin{equation}
(\mathring{D}_{\boldsymbol{v}}\mathcal{C)}_{\perp}:=\mathbf{P}_{\perp}%
(v\cdot\boldsymbol{\mathring{\partial}}\mathcal{C}) \label{ss4}%
\end{equation}
we have%
\begin{equation}
v\cdot\boldsymbol{\partial}=\mathbf{P(}v\cdot\boldsymbol{\mathring{\partial}%
}\mathbf{)}=(v\cdot\boldsymbol{\mathring{\partial}})_{\parallel}%
=(v\cdot\mathfrak{\mathring{d}})_{\parallel}=v\cdot\boldsymbol{\mathring
{\partial}}-\mathbf{P}_{\perp}(v\cdot\boldsymbol{\mathring{\partial}}%
)=v\cdot\mathfrak{\mathring{d}-}\mathbf{P}_{\perp}(v\cdot\mathfrak{\mathring
{d}}). \label{p5}%
\end{equation}
So, it is
\begin{align}
v\cdot\mathfrak{\mathring{d}}I_{m}I_{m}^{-1}  &  =%
{\textstyle\sum\nolimits_{j=1}^{m}}
\theta^{1}\cdots(D_{\boldsymbol{v}}\theta^{j}+\mathbf{P}_{\perp}%
(v\cdot\mathfrak{\mathring{d}}\theta\mathfrak{\mathfrak{^{j}}))\cdots}%
\theta^{m}I_{m}^{-1}\nonumber\\
&  =D_{\boldsymbol{v}}I_{m}I_{m}^{-1}+\mathbf{P}_{\perp}(v\cdot
\mathfrak{\mathring{d}\theta}_{j}\mathfrak{)\wedge\theta^{j}.} \label{p5a}%
\end{align}

Now, $D_{\boldsymbol{v}}I_{m}\in\sec%
{\textstyle\bigwedge\nolimits^{m}}
T^{\ast}M\hookrightarrow\sec\mathcal{C}\ell(M,\mathtt{g})$ is a multiple of
$I_{m}$ and since $I_{m}^{2}=\pm1$ depending on the signature of the metric
$\boldsymbol{g}$ we have that $D_{\boldsymbol{v}}I_{m}=0$.

Indeed,
\begin{equation}
0=D_{\boldsymbol{v}}I_{m}^{2}=2(D_{\boldsymbol{v}}I_{m})I_{m} \label{p7}%
\end{equation}
and so%
\[
0=(D_{\boldsymbol{v}}I_{m})I_{m}I_{m}^{-1}=D_{\boldsymbol{v}}I_{m}.
\]
In any Clifford algebra bundle, in particular $\mathcal{C}\ell(M,\mathtt{g})$
we can build multiples of the $I_{m}$ not only multiplying it by an scalar
function, but also multiplying it by a an appropriated biform. This result
will be used below to define the shape biform.

\subsection{$\mathbf{S}(\mathring{v})=\mathbf{S}(v)=\partial_{u}%
\boldsymbol{\wedge}\mathbf{P}_{u}\boldsymbol{~}(\mathring{v})$ and
$\mathbf{S}(\mathring{v}_{\perp})=\partial_{u}\boldsymbol{\lrcorner}%
\mathbf{P}_{u}\boldsymbol{~}(\mathring{v})$}

For any $\mathcal{C\in}\sec$ $\mathcal{C}\ell(M,\mathtt{g})$ it is
$\mathcal{C}=\mathbf{P}(\mathcal{C})$ we have (with $u\in\sec%
{\textstyle\bigwedge\nolimits^{1}}
T^{\ast}\mathring{M}\hookrightarrow\sec$ $\mathcal{C}\ell(\mathring
{M},\mathtt{\mathring{g}})$)
\begin{align}
\mathfrak{\mathring{d}}\boldsymbol{(}\mathbf{P}(\mathcal{C}))  &
=\mathfrak{\mathring{d}}\mathbf{P}\boldsymbol{~}(\mathcal{C})-\mathbf{P}%
(\mathfrak{\mathring{d}}\mathcal{C})\nonumber\\
&  =\partial_{u}\mathbf{P}_{u}\boldsymbol{~}(\mathcal{C})-\mathbf{P}%
(\mathfrak{\mathring{d}}\mathcal{C})\nonumber\\
&  =\partial_{u}\boldsymbol{\wedge}\mathbf{P}_{u}\boldsymbol{~}(\mathcal{C}%
)+\partial_{u}\boldsymbol{\lrcorner}\mathbf{P}_{u}\boldsymbol{~}%
(\mathcal{C})-\mathbf{P}(\mathfrak{\mathring{d}}\mathcal{C}) \label{A1}%
\end{align}
where
\begin{equation}
\mathbf{P}_{u}\boldsymbol{~}(\mathcal{C}):=u\cdot\mathfrak{\mathring{d}%
}\mathbf{P~}(\mathcal{C})=u\cdot\mathfrak{\mathring{d}(}\mathbf{P}%
(\mathcal{C}))-\mathfrak{(}\mathbf{P}(u\cdot\mathfrak{\mathring{d}}%
\mathcal{C})). \label{A1a}%
\end{equation}

Recall that for $\mathring{v}\in\sec%
{\textstyle\bigwedge\nolimits^{1}}
T^{\ast}\mathring{M}\hookrightarrow\sec$ $\mathcal{C}\ell(\mathring
{M},\mathtt{\mathring{g}})$ we can write%
\begin{equation}
\mathbf{S}(\mathring{v})=\mathfrak{\mathring{d}}\mathbf{P}\boldsymbol{~}%
(\mathring{v})=\partial_{u}\boldsymbol{\wedge}\mathbf{P}_{u}\boldsymbol{~}%
(\mathring{v})+\partial_{u}\boldsymbol{\lrcorner}\mathbf{P}_{u}\boldsymbol{~}%
(\mathring{v}) \label{A2}%
\end{equation}
where we used that for any $\mathcal{\mathring{C}\in}\sec$ $\mathcal{C}%
\ell(\mathring{M},\mathtt{\mathring{g}})$ it is%
\begin{equation}
\partial_{u}\mathbf{P}_{u}\boldsymbol{~}(\mathcal{\mathring{C}}):=%
{\textstyle\sum\nolimits_{i=1}^{m}}
\theta^{\mathbf{i}}\frac{\partial}{\partial u^{\mathbf{i}}}u\cdot
\mathfrak{\mathring{d}}\mathbf{P\boldsymbol{~}}(\mathcal{\mathring{C}})=%
{\textstyle\sum\nolimits_{i=1}^{m}}
\theta^{\mathbf{i}}\mathring{D}_{\boldsymbol{\mathring{e}}_{\mathbf{i}}%
}\mathbf{P\boldsymbol{~}}(\mathcal{\mathring{C}})=\mathfrak{\mathring{d}%
}\mathbf{P\boldsymbol{~}}(\mathcal{\mathring{C}}) \label{A3}%
\end{equation}

Putting $\mathring{v}=\mathring{v}_{\parallel}+\mathring{v}_{\perp
}=v+\mathring{v}_{\perp}$ we have the

\begin{proposition}
\label{proS}%
\begin{equation}
\mathbf{S}(\mathring{v})=\mathbf{S}(v)=\partial_{u}\boldsymbol{\wedge
}\mathbf{P}_{u}\boldsymbol{~}(\mathring{v}),~~~~~~\mathbf{S}(\mathring
{v}_{\perp})=\partial_{u}\lrcorner\mathbf{P}_{u}(\mathring{v}). \label{A4}%
\end{equation}

\end{proposition}

\begin{proof}
Indeed,%

\begin{equation}
\mathbf{S}(\mathring{v})=\mathbf{S}(\mathring{v}_{\parallel})+\mathbf{S}%
(\mathring{v}_{\perp})=\partial_{u}\boldsymbol{\wedge}\mathbf{P}%
_{u}\boldsymbol{~}(\mathring{v})+\partial_{u}\boldsymbol{\lrcorner}%
\mathbf{P}_{u}\boldsymbol{~}(\mathring{v}) \label{A5a}%
\end{equation}
and so, it is enough to show that
\begin{equation}
\partial_{u}\boldsymbol{\lrcorner}\mathbf{P}_{u}\boldsymbol{~}(\mathring
{v}_{\parallel})=0\text{ \ \ and \ \ }\partial_{u}\boldsymbol{\wedge
}\mathbf{P}_{u}\boldsymbol{~}(\mathring{v}_{\perp})=0, \label{A6}%
\end{equation}
From $\mathbf{P}^{2}(\mathring{v})=\mathbf{P}(\mathring{v})$ we get
\begin{equation}
\mathbf{P}_{u}\mathbf{P}(\mathring{v}_{\parallel})+\mathbf{PP}_{u}%
(\mathring{v}_{\parallel})=\mathbf{P}_{u}(\mathring{v}_{\parallel}),
\label{A7}%
\end{equation}
So, for $\mathring{v}_{\parallel}$ and $\mathring{v}_{\perp}$ it is
\begin{equation}
\mathbf{PP}_{u}(\mathring{v}_{\parallel})=0\text{ \ \ \ and \ \ }%
\mathbf{PP}_{u}(\mathring{v}_{\perp})=\mathbf{P}_{u}(\mathring{v}_{\perp}).
\label{A8}%
\end{equation}
Since $\mathbf{PP}_{u}(\mathring{v}_{\parallel})=0$ we have that%
\begin{equation}
\partial_{u}\lrcorner\mathbf{PP}_{u}(\mathring{v}_{\parallel})=\mathbf{P(}%
\partial_{u})\lrcorner\mathbf{P}_{u}(\mathring{v}_{\parallel})=\partial
_{u}\lrcorner\mathbf{P}_{u}(\mathring{v}_{\parallel})=0. \label{A9}%
\end{equation}
\ \ From $\mathbf{PP}_{u}(\mathring{v}_{\perp})=\mathbf{P}_{u}(\mathring
{v}_{\perp})$ we can write%
\begin{equation}
\partial_{u}\wedge\mathbf{P}_{u}(\mathring{v}_{\perp})=\partial_{u}%
\wedge\mathbf{PP}_{u}(\mathring{v}_{\perp})=\mathbf{P}(\partial_{u}%
)\wedge\mathbf{PP}_{u}(\mathring{v}_{\perp})=\mathbf{P(}\partial_{u}%
\wedge\mathbf{P}_{u}(\mathring{v}_{\perp})). \label{A10}%
\end{equation}
Now, take $t,y\in\sec%
{\textstyle\bigwedge\nolimits^{1}}
T^{\ast}M\hookrightarrow\sec$ $\mathcal{C}\ell(M,\mathtt{g})$. We have
\begin{gather}
(t\wedge y)\cdot(\partial_{u}\wedge\mathbf{PP}_{u}(\mathring{v}_{\perp
}))=(t\wedge y)\cdot\mathbf{(}\partial_{u}\wedge\mathbf{P}_{u}(\mathring
{v}_{\perp}))\nonumber\\
=t\lrcorner((y\cdot\boldsymbol{\partial}_{u}\wedge\mathbf{P}_{u}(\mathring
{v}_{\perp})-\boldsymbol{\partial}_{u}\wedge((y\lrcorner\mathbf{P}%
_{u}(\mathring{v}_{\perp}))\nonumber\\
=t\lrcorner(y\cdot\mathfrak{\mathring{d}(}\mathbf{P}(\mathring{v}_{\perp
}))-\mathbf{P(}y\cdot\mathfrak{\mathring{d}}\mathring{v}_{\perp}%
)-\boldsymbol{\theta}^{\mathbf{i}}\wedge(y\lrcorner(\boldsymbol{\theta
}_{\mathbf{i}}\cdot\mathfrak{\mathring{d}}(\mathbf{P}(\mathring{v}_{\perp
}))-\mathbf{P(}\boldsymbol{\theta}_{\mathbf{i}}\cdot\mathfrak{\mathring{d}%
}\mathring{v}_{\perp})))\nonumber\\
=t\lrcorner(\mathbf{P(}y\cdot\mathfrak{\mathring{d}}\mathring{v}_{\perp
})-\boldsymbol{\theta}^{\mathbf{i}}\wedge(y\lrcorner\mathbf{P(}%
\boldsymbol{\theta}_{\mathbf{i}}\cdot\mathfrak{\mathring{d}}\mathring
{v}_{\perp})))\nonumber\\
=t\lrcorner(D_{\boldsymbol{y}}\mathring{v}_{\perp})-\boldsymbol{\theta
}^{\mathbf{i}}\wedge(y\lrcorner(D_{\boldsymbol{e}_{i}}\mathring{v}_{\perp
})))=0 \label{A11}%
\end{gather}
from where it follows that%
\begin{equation}
\partial_{u}\wedge\mathbf{P}_{u}(\mathring{v}_{\perp})=0 \label{A12}%
\end{equation}
and the proposition is proved.
\end{proof}

\subsubsection{Shape Biform $\mathcal{S}$}

\begin{definition}
The shape biform \emph{(}a $(1,2)$-extensor field\emph{)} is the mapping
\begin{gather}
\mathcal{S}:\sec%
{\textstyle\bigwedge\nolimits^{1}}
T^{\ast}M\rightarrow%
{\textstyle\bigwedge\nolimits^{2}}
T^{\ast}M,\label{p6a}\\
v\boldsymbol{\mapsto}\mathcal{S(}v),\nonumber
\end{gather}
such that
\begin{equation}
v\cdot\mathfrak{\mathring{d}}I_{m}=-\mathcal{S(}v)I_{m}. \label{p6b}%
\end{equation}

\end{definition}

From Eq.(\ref{pcomu}) it follows that%
\begin{equation}
\mathcal{S(}v)\lrcorner I_{m}=0\text{ and }\mathcal{S(}v)\wedge I_{m}=0.
\label{p8}%
\end{equation}
Since $\mathcal{S(}v)\lrcorner I_{m}=0$ it follows from Remark \ref{outside}
that
\begin{equation}
\mathbf{P(}\mathcal{S(}v))=0. \label{P9}%
\end{equation}
Now, using the fact that $D_{\boldsymbol{v}}I_{m}=0$ and Eq.(\ref{p6b}) it
follows from Eq.(\ref{p5a}) that
\[
v\cdot\mathfrak{\mathring{d}}I_{m}=(\mathbf{P}_{\perp}(v\cdot
\mathfrak{\mathring{d}\theta}_{j}\mathfrak{)\wedge\theta^{j})}I_{m}%
=-\mathcal{S(}v)I_{m},
\]
i.e.,
\begin{equation}
\mathcal{S(}v)=-\mathbf{P}_{\perp}(v\cdot\mathfrak{\mathring{d}}\theta
_{j}\mathfrak{)\wedge\theta^{j}}. \label{p10}%
\end{equation}

\begin{proposition}
For any $\mathcal{C}\in\sec\mathcal{C}\ell(M,\mathtt{g})$ we have
\begin{equation}
D_{\boldsymbol{v}}\mathcal{C}=v\cdot\mathfrak{\mathring{d}}\mathcal{C}%
+\mathcal{S(}v)\times\mathcal{C}=\mathring{D}_{\boldsymbol{v}}\mathcal{C}%
+\mathcal{S(}v)\times\mathcal{C}. \label{p11}%
\end{equation}

\end{proposition}

\begin{proof}
Taking into account Eq.(\ref{p10}) we have for any $v,w\in\sec%
{\textstyle\bigwedge\nolimits^{1}}
T^{\ast}M\hookrightarrow\sec\mathcal{C}\ell(M,\mathtt{g})$%
\begin{align}
v\boldsymbol{\lrcorner}\mathcal{S}(w)  &  =-v\lrcorner(\mathbf{P}_{\perp
}(w\cdot\mathfrak{\mathring{d}}\theta_{j}\mathfrak{)\wedge}\theta
\mathfrak{^{j})}\nonumber\\
&  =-(\boldsymbol{v}\lrcorner(\mathbf{P}_{\perp}(w\cdot\mathfrak{\mathring{d}%
}\theta_{j}\mathfrak{))}\theta\mathfrak{^{j}}+v^{j}\mathbf{P}_{\perp}%
(w\cdot\mathfrak{\mathring{d}}\theta_{j}\mathfrak{)}\nonumber\\
&  =v^{j}\mathbf{P}_{\perp}(w\cdot\mathfrak{\mathring{d}}\theta_{j}%
\mathfrak{)}=\mathbf{P}_{\perp}(w\cdot\mathfrak{\mathring{d}}v)-\mathbf{P}%
_{\perp}[(w\cdot\mathfrak{\mathring{d}}v^{j})\theta_{j}\mathfrak{]}\nonumber\\
&  =\mathbf{P}_{\perp}(w\cdot\mathfrak{\mathring{d}}v). \label{p13}%
\end{align}
So,%

\begin{align}
v\cdot\mathfrak{\mathring{d}}w  &  =\mathbf{P(}v\mathbf{\cdot}%
\mathfrak{\mathring{d}}w\mathbf{)+P}_{\perp}(v\cdot\mathfrak{\mathring{d}%
}w)\nonumber\\
&  =D_{\boldsymbol{v}}w+w\boldsymbol{\lrcorner}\mathcal{S}(v)\nonumber\\
&  =D_{\boldsymbol{v}}w-\mathcal{S}(v)\llcorner w\boldsymbol{.} \label{p14}%
\end{align}
Now, for $v,w\in\sec%
{\textstyle\bigwedge\nolimits^{1}}
T^{\ast}M\hookrightarrow\sec\mathcal{C}\ell(M,\mathtt{g})$ we have%
\begin{align}
D_{\boldsymbol{v}}(wu)  &  =(D_{\boldsymbol{v}}w)u+wD_{\boldsymbol{v}%
}u=(\mathring{D}_{\boldsymbol{v}}w)u+(\mathcal{S}(v)\times w)u+w\mathring
{D}_{\boldsymbol{v}}u+w(\mathcal{S(}v)\times u)\nonumber\\
&  =(\mathring{D}_{\boldsymbol{v}}wu)+(\mathcal{S}(v)\times w)u-w(u\times
\mathcal{S(}v))\nonumber\\
&  =(\mathring{D}_{\boldsymbol{v}}wu)+(\mathcal{S}(v)\times wu), \label{p15}%
\end{align}
from where the proposition follows trivially by finite induction.
\end{proof}

Of course, it is
\begin{equation}
D_{\boldsymbol{e}_{_{i}}}\mathcal{C}=\mathring{D}_{\boldsymbol{e}_{\mathbf{i}%
}}\mathcal{C}+\mathcal{S(}v)\times\mathcal{C} \label{p16}%
\end{equation}
Now, recalling Eq.(\ref{code}) we have\footnote{Take notice that this formula
being gauge dependent is not valid if $\boldsymbol{e}_{i}\mapsto
\mathbf{x}_{\mathbf{i}}$ where the $\mathbf{x}_{\mathbf{i}}$ coordinate vector
fields. See Corollary \ref{coro1}.}
\begin{equation}
\mathring{D}_{\boldsymbol{e}_{\mathbf{i}}}\mathcal{C}=\eth_{\boldsymbol{e}%
_{\mathbf{i}}}\mathcal{C}+\mathring{\omega}_{e_{i}}\times\mathcal{C}
\label{p17}%
\end{equation}
where for $\mathbf{i},\mathbf{j}=1,...,m,$ $\mathring{D}_{\boldsymbol{e}%
_{\mathbf{i}}}\theta^{\mathbf{j}}=\mathring{D}_{\boldsymbol{e}_{i}}%
\mathring{\theta}^{j}=-%
{\textstyle\sum\nolimits_{\mathbf{k}=1}^{n}}
\mathring{\omega}_{\cdot\mathbf{ik}}^{\mathbf{j}\cdot\cdot}\mathring{\theta
}^{\mathbf{k}}$, it is
\begin{equation}
\mathring{\omega}_{\boldsymbol{v}}=\frac{1}{2}v^{\mathbf{c}}\mathring{\omega
}_{\cdot\mathbf{c}\cdot}^{\mathbf{a}\cdot\mathbf{b}}\theta_{\mathbf{a}}%
\wedge\theta_{\mathbf{b}} \label{p18}%
\end{equation}

So, we get
\begin{align}
D_{\boldsymbol{v}}\mathcal{C}  &  =v\cdot\mathfrak{\mathring{d}}%
\mathcal{C}+\mathcal{S(}v)\times\mathcal{C}\nonumber\\
&  =\mathfrak{d}_{\boldsymbol{v}}\mathcal{C}+\mathfrak{(}\mathring{\omega
}_{\boldsymbol{v}}+\mathcal{S(}v))\times\mathcal{C} \label{p19}%
\end{align}
and in particular%
\begin{equation}
D_{\boldsymbol{e}_{_{i}}}\mathcal{C}=\eth_{\boldsymbol{e}_{\mathbf{i}}%
}\mathcal{C}+(\mathring{\omega}_{e_{i}}+\mathcal{S(}\boldsymbol{e}_{_{i}%
}))\times\mathcal{C}. \label{p20}%
\end{equation}
Comparison of Eq.(\ref{p20}) with Eq.(\ref{om2}) (valid for any metric
compatible connection) implies the important result%
\begin{equation}
\omega_{\boldsymbol{v}}=\mathfrak{(}\mathring{\omega}_{\boldsymbol{v}%
}+\mathcal{S(}v)) \label{p21}%
\end{equation}
We can easily find by direct calculation that in a gauge where $\mathring
{\omega}_{\boldsymbol{v}}\neq0$,
\begin{equation}
\omega_{\boldsymbol{v}}=\mathbf{P}\mathfrak{(}\mathring{\omega}%
_{\boldsymbol{v}}) \label{p21a}%
\end{equation}
which is consistent with the fact that from Eq.(\ref{P9}) it is $\mathbf{P}%
(\mathcal{S(}v))=0$.

Let $(\boldsymbol{x}^{1},...,\boldsymbol{x}^{m},...\boldsymbol{x}^{n})$ be the
natural orthogonal coordinate functions of $\mathring{M}\simeq\mathbb{R}^{n}$

\begin{corollary}
\label{coro1}For $\mathcal{C}\in\sec$ $\mathcal{C}\ell(M,\mathtt{g})$
\begin{equation}
D_{\boldsymbol{v}}\mathcal{C}=v^{\mathbf{i}}\frac{\partial}{\partial
\boldsymbol{x}^{\mathbf{i}}}\mathcal{C}+\mathcal{S(}v)\times\mathcal{C}.
\label{cor1}%
\end{equation}

\end{corollary}

\begin{proof}
Taking into account that $D_{\frac{\partial}{\partial\boldsymbol{x}^{i}}%
}d\boldsymbol{x}^{j}=0$ follows that $\mathcal{\mathring{\omega}}%
_{\frac{\partial}{\partial\boldsymbol{x}^{\mathbf{i}}}}=\frac{1}{2}%
(\mathring{\Gamma}_{\mathbf{kil}})d\boldsymbol{x}^{k}\wedge d\boldsymbol{x}%
^{\mathbf{l}}=0$. Using this result in Eq.(\ref{p11}) with $\boldsymbol{e}%
_{\mathbf{i}}\mapsto\frac{\partial}{\partial\boldsymbol{x}^{\mathbf{i}}}$
gives the desired result.
\end{proof}

\begin{remark}
Comparison of \emph{Eq.(\ref{p19})} and \emph{Eq.(\ref{cor1})} shows that
$\mathcal{S(}v)$ cannot always be identified with $\omega\mathcal{(}v)$ which
is a gauge dependent operator.
\end{remark}

\subsection{Integrability Conditions}

\begin{remark}
Take into account that the commutator of Pfaff derivatives acting on any
$\mathcal{C}\in\sec$ $\mathcal{C}\ell(M,\mathtt{g})$ is in general non null,
i.e.,
\begin{equation}
\lbrack\eth_{\boldsymbol{e}_{\mathbf{i}}},\eth_{\boldsymbol{e}_{\mathbf{j}}%
}]\mathcal{C}=%
{\textstyle\sum\nolimits_{r=0}^{m}}
c_{\cdot\mathbf{ij}}^{\mathbf{k\cdot\cdot}}\boldsymbol{e}_{\mathbf{k}%
}(\mathcal{C}_{\mathbf{i}_{1}\cdots\mathbf{i}_{m}})\theta^{\mathbf{i}%
_{1}\cdots\mathbf{i}_{m}}\neq0, \label{p22}%
\end{equation}
unless $\boldsymbol{e}_{\mathbf{i}}$ are coordinate vector fields, i.e.,
$\boldsymbol{e}_{\mathbf{i}}\mapsto\mathbf{x}_{\mathbf{i}}$.
\end{remark}

\begin{remark}
Also, since the torsion of $\mathring{D}$ is null we have in general
\begin{equation}
\lbrack\theta_{\mathbf{i}}\cdot\mathfrak{\mathring{d},}\theta_{\mathbf{j}%
}\cdot\mathfrak{\mathring{d}]}\mathcal{C=[}\theta_{\mathbf{i}},\theta
_{\mathbf{j}}]\cdot\mathfrak{\mathring{d}}\mathcal{C}=c_{\cdot\mathbf{ij}%
}^{\mathbf{k\cdot\cdot}}\theta_{\mathbf{k}}\cdot\mathfrak{\mathring{d}%
}\mathcal{C}=c_{\cdot\mathbf{ij}}^{\mathbf{k\cdot\cdot}}\mathring
{D}_{\boldsymbol{e}_{\mathbf{k}}}\mathcal{C\neq}0, \label{p23}%
\end{equation}
unless $\boldsymbol{e}_{\mathbf{i}}$ are coordinate vector fields. Moreover,
for the case of orthonormal vector fields
\begin{equation}
\lbrack\theta_{\mathbf{i}}\cdot\mathfrak{\mathring{d},}\theta_{\mathbf{j}%
}\cdot\mathfrak{\mathring{d}]}\mathcal{C}\neq\lbrack\eth_{\boldsymbol{e}%
_{\mathbf{i}}},\eth_{\boldsymbol{e}_{\mathbf{j}}}]\mathcal{C}. \label{p23a}%
\end{equation}

\end{remark}

\begin{remark}
The integrability condition for the connection $\mathring{D}$ is expressed,
given the previous results, by
\begin{equation}
\boldsymbol{\mathring{\partial}}\wedge\boldsymbol{\mathring{\partial}}=0
\label{p24}%
\end{equation}
which means that for any $\mathcal{\mathring{C}}\in\sec$ $\mathcal{C}%
\ell(\mathring{M},\mathtt{\mathring{g}})$ it is
\[
\boldsymbol{\mathring{\partial}}\wedge\boldsymbol{\mathring{\partial}%
~}\mathcal{\mathring{C}}=0
\]
For the manifold $M$ recalling that $\mathbf{x}_{\mathbf{i}}\equiv
\eth_{\mathbf{x}_{\mathbf{i}}}$ is a Pfaff derivative we have for any
$\mathcal{C\in}\sec$ $\mathcal{C}\ell(M,\mathtt{g})$%
\begin{equation}
(\eth_{\mathbf{x}_{\mathbf{i}}}\eth_{\mathbf{xj}}-\eth_{\mathbf{x}%
_{\mathbf{j}}}\eth_{\mathbf{x}_{\mathbf{i}}})\mathcal{C}=0. \label{p25}%
\end{equation}
If we recall the definition of the form derivative \emph{(Eq.(\ref{formderiv1}%
)), putting}%
\begin{equation}
\eth:=\vartheta^{\mathbf{i}}\eth_{\mathbf{x}_{\mathbf{i}}} \label{p26}%
\end{equation}
we can express the `integrability' condition in $M$ by
\begin{equation}
\eth\mathfrak{\wedge}\eth=0. \label{p27}%
\end{equation}
Finally recalling \emph{Eqs.(\ref{ri1}) and (\ref{ri2})} for $v\in\sec%
{\textstyle\bigwedge\nolimits^{1}}
T^{\ast}M\hookrightarrow\sec$ $\mathcal{C}\ell(M,\mathtt{g})$ it is
\begin{equation}
\boldsymbol{\partial}\wedge\boldsymbol{\partial~}v=v_{\mathbf{i}}%
\mathcal{R}^{\mathbf{i}} \label{p28}%
\end{equation}
where $\mathcal{R}^{\mathbf{i}}$ are the Ricci $1$-form fields.
\end{remark}

\subsection{$\mathbf{S}(v)=\mathcal{S(}v)$}

\begin{proposition}
Let $\mathcal{C}=v\in\sec%
{\textstyle\bigwedge\nolimits^{1}}
T^{\ast}M\hookrightarrow\sec$ $\mathcal{C}\ell(M,\mathtt{g})$ we have%
\begin{equation}
\mathbf{S}(v)=\mathcal{S(}v). \label{S2}%
\end{equation}

\end{proposition}

\begin{proof}
We have%
\begin{align}
\mathbf{S}(v)  &  =\mathfrak{\mathring{d}(}\mathbf{P}(v))-\mathbf{P}%
\mathfrak{(\mathring{d}}v)\nonumber\\
&  =\mathfrak{\mathring{d}}v-\mathbf{P}\mathfrak{(\mathring{d}}v). \label{S3}%
\end{align}
Now, $\mathfrak{\mathring{d}}v=\mathfrak{\mathring{d}\wedge}%
v+\mathfrak{\mathring{d}\lrcorner}v$ and since $\mathbf{P}\mathfrak{(\mathring
{d}\lrcorner}v)=\mathfrak{\mathring{d}\lrcorner}v$ we have
\begin{equation}
\mathbf{S}(v)=\mathfrak{\mathring{d}\wedge}v-\mathbf{P}\mathfrak{(\mathring
{d}\wedge}v) \label{S4}%
\end{equation}
\ It is only necessary due to the linearity $\mathbf{S}$ of to show
Eq.(\ref{S2}) for $v=\theta_{\mathbf{d}}$, $\mathbf{d=}1,...,m$. We then
evaluate%
\begin{align}
\mathfrak{\mathring{d}\wedge}\theta_{\mathbf{d}}  &  =%
{\textstyle\sum\nolimits_{\mathbf{k}=1}^{m}}
\theta^{\mathbf{k}}\mathring{D}_{\boldsymbol{e}_{\mathbf{k}}}\theta
_{\mathbf{d}}=%
{\textstyle\sum\nolimits_{\mathbf{k}=1}^{m}}
{\textstyle\sum\nolimits_{t=1,t\neq k}^{n}}
\mathring{\omega}_{t\mathbf{kd}}\theta^{\mathbf{k}}\wedge\mathring{\theta}%
^{t}\nonumber\\
&  =%
{\textstyle\sum\nolimits_{k=1}^{m}}
{\textstyle\sum\nolimits_{t=1,t\neq k}^{m}}
\mathring{\omega}_{t\mathbf{kd}}\theta^{\mathbf{k}}\wedge\theta^{\mathbf{t}}+%
{\textstyle\sum\nolimits_{\mathbf{k}=1}^{m}}
{\textstyle\sum\nolimits_{t=m+1}^{m+l}}
\mathring{\omega}_{t\mathbf{kd}}\theta^{\mathbf{k}}\wedge\mathring{\theta}%
^{t}, \label{s5}%
\end{align}
from where it follows that%
\begin{equation}
\mathbf{S}(\theta_{\mathbf{d}})=%
{\textstyle\sum\nolimits_{\mathbf{k}=1}^{m}}
{\textstyle\sum\nolimits_{t=m+1}^{m+l}}
\mathring{\omega}_{t\mathbf{kd}}\theta^{\mathbf{k}}\wedge\mathring{\theta}%
^{t}=\frac{1}{2}%
{\textstyle\sum\nolimits_{\mathbf{k}=1}^{m}}
{\textstyle\sum\nolimits_{t=m+1}^{m+l}}
(\mathring{\omega}_{t\mathbf{kd}}-\mathring{\omega}_{\mathbf{k}t\mathbf{d}%
})\theta^{\mathbf{k}}\wedge\mathring{\theta}^{t} \label{S6}%
\end{equation}
On the other hand%
\begin{align}
\theta_{\mathbf{d}}\cdot\mathfrak{\mathring{d}}I_{m}  &  =\eta^{11}\cdots
\eta^{mm}\mathring{D}_{\boldsymbol{e}_{\mathbf{d}}}(\theta_{1}\wedge
\cdots\wedge\theta_{m})\nonumber\\
&  =\alpha\mathring{D}_{\boldsymbol{e}_{\mathbf{d}}}(\theta_{1}\cdots
\theta_{m})\nonumber\\
&  =\alpha%
{\textstyle\sum\nolimits_{k=1}^{m}}
{\textstyle\sum\nolimits_{t=1}^{n}}
\mathring{\omega}_{t\mathbf{dk}}\theta_{1}\cdots\underset{k\text{-position}%
}{\underbrace{\mathring{\theta}^{t}}}\cdots\theta_{m}\nonumber\\
&  =\alpha%
{\textstyle\sum\nolimits_{k=1}^{m}}
{\textstyle\sum\nolimits_{t=1}^{m}}
\mathring{\omega}_{t\mathbf{dk}}\theta_{1}\cdots\underset{k\text{-position}%
}{\underbrace{\theta^{t}}}\cdots\theta_{m}\nonumber\\
&  +\alpha%
{\textstyle\sum\nolimits_{k=1}^{m}}
{\textstyle\sum\nolimits_{t=m+1}^{m+l}}
\mathring{\omega}_{t\mathbf{dk}}\theta_{1}\cdots\underset{k\text{-position}%
}{\underbrace{\mathring{\theta}^{t}}}\cdots\theta_{m} \label{S7}%
\end{align}
and now we can easily see that%
\begin{equation}
\mathbf{S}(\theta_{\mathbf{d}})\times I_{m}=\theta_{\mathbf{d}}\cdot
\mathfrak{\mathring{d}}I_{m} \label{S8}%
\end{equation}
and it follows that $\mathbf{S}(\theta_{\mathbf{d}})=\mathcal{S(}%
\theta_{\mathbf{d}}\mathcal{)}$.
\end{proof}

We also have the

\begin{proposition}
Let $v,w\in\sec%
{\textstyle\bigwedge\nolimits^{1}}
T^{\ast}M\hookrightarrow\sec$ $\mathcal{C}\ell(M,\mathtt{g})$ we have%

\begin{equation}
v\cdot\mathcal{S(}w\mathcal{)}=w\cdot\mathcal{S(}v\mathcal{)} \label{s9}%
\end{equation}

\end{proposition}

\begin{proof}
Recalling Eq.(\ref{S6}) we can write%
\begin{align*}
v\cdot\mathcal{S(}w\mathcal{)}  &  =%
{\textstyle\sum\nolimits_{\mathbf{i,d}=1}^{m}}
v^{\mathbf{i}}w^{\mathbf{d}}\theta_{\mathbf{i}}\lrcorner%
{\textstyle\sum\nolimits_{\mathbf{k}=1}^{m}}
{\textstyle\sum\nolimits_{t=m+1}^{m+l}}
\mathring{\omega}_{t\mathbf{kd}}\theta^{\mathbf{k}}\wedge\mathring{\theta}%
^{t}\\
&  =\frac{1}{2}%
{\textstyle\sum\nolimits_{\mathbf{i,d}=1}^{m}}
v^{\mathbf{i}}w^{\mathbf{d}}(\mathring{\omega}_{t\mathbf{id}}-\mathring
{\omega}_{\mathbf{i}t\mathbf{d}})\mathring{\theta}^{t}=w\cdot\mathcal{S(}%
v\mathcal{)}%
\end{align*}
and the proposition is proved.
\end{proof}

\subsection{$\mathfrak{\mathring{d}}\wedge v=$ $\boldsymbol{\partial}\wedge
v+\mathcal{S}(v)$ and $\mathfrak{\mathring{d}}\boldsymbol{\lrcorner
}v=\boldsymbol{\partial\lrcorner v}$}

We first observe that since torsion is null for the Levi-Civita connection
$\mathring{D}$ we have for any $u,v\in\sec T^{\ast}\mathring{M}\hookrightarrow
\sec\mathcal{C\ell}(\mathring{M},\mathtt{\mathring{g}})$ we have
\begin{equation}
u\cdot\boldsymbol{\mathring{\partial}}v=v\cdot\boldsymbol{\mathring{\partial}%
}u+%
\bra
u,v%
\ket
\label{O1}%
\end{equation}
from where it follows $\mathbf{P}_{\perp}(%
\bra
u,v%
\ket
)=0$\ when $u,v\in\sec T^{\ast}M\hookrightarrow\sec\mathcal{C\ell
}(M,\mathtt{g})$ since calculating $%
\bra
u,v%
\ket
$ with $\boldsymbol{\mathring{\partial}}$ expressed in the natural coordinates
of $\mathring{M}$ we find that $[u,v]\in\sec T^{\ast}M\hookrightarrow
\sec\mathcal{C\ell}(M,\mathtt{g})$. From this it follows that
\begin{equation}
\mathbf{P}_{\perp}(u\cdot\mathfrak{\mathring{d}}v)=\mathbf{P}_{\perp}%
(v\cdot\mathfrak{\mathring{d}}u). \label{O2}%
\end{equation}

Then we can write%
\begin{align*}
\mathfrak{\mathring{d}}\wedge v  &  =%
{\textstyle\sum\limits_{\mathbf{m,k=1}}^{m}}
\theta^{\mathbf{r}}\wedge\mathring{D}_{\boldsymbol{e}_{\mathbf{r}}%
}(v^{\mathbf{k}}\theta_{\mathbf{k}})\\
&  =%
{\textstyle\sum\limits_{\mathbf{r,k=1}}^{m}}
\theta^{\mathbf{r}}\wedge\{\boldsymbol{e}_{\mathbf{r}}(v^{\mathbf{k}%
})+v^{\mathbf{k}}%
{\textstyle\sum\limits_{\mathbf{s=1}}^{\mathbf{m}}}
L_{\mathbf{rk}}^{\mathbf{s}}\theta_{\mathbf{s}})\}+%
{\textstyle\sum\limits_{\mathbf{r,k=1}}^{m}}
\theta^{\mathbf{r}}\wedge v^{\mathbf{k}}%
{\textstyle\sum\limits_{\mathbf{s=m+1}}^{\mathbf{m+l=n}}}
L_{\mathbf{rk}}^{\mathbf{s}}\mathring{\theta}_{\mathbf{s}}\\
&  =\boldsymbol{\partial}\wedge v+%
{\textstyle\sum\limits_{\mathbf{m,k=1}}^{m}}
\theta^{\mathbf{r}}\wedge v^{\mathbf{k}}\mathbf{P}_{\perp}(\mathring
{D}_{\mathbf{r}}\mathfrak{\theta}_{\mathbf{k}})\\
&  =\boldsymbol{\partial}\wedge v+%
{\textstyle\sum\limits_{\mathbf{m,k=1}}^{m}}
\theta^{\mathbf{r}}\wedge v^{\mathbf{k}}\mathbf{P}_{\perp}(\mathring
{D}_{\mathbf{k}}\mathfrak{\theta}_{\mathbf{r}})\\
&  =\boldsymbol{\partial}\wedge v+\theta^{\mathbf{r}}\wedge\mathbf{P}_{\perp
}(\mathring{D}_{\mathbf{v}}\mathfrak{\theta}_{\mathbf{r}})\\
&  =\boldsymbol{\partial}\wedge v+\theta^{\mathbf{r}}\wedge\mathbf{P}_{\perp
}(v\cdot\mathfrak{\mathring{d}\theta}_{\mathbf{r}})
\end{align*}
and recalling that $\mathcal{S(}v)=-\mathbf{P}_{\perp}(v\cdot
\mathfrak{\mathring{d}}\theta_{\mathbf{r}}\mathfrak{)\wedge}\theta
^{\mathbf{r}}$ \ we finally have
\begin{equation}
\mathfrak{\mathring{d}}\wedge v=\boldsymbol{\partial}\wedge v+\mathcal{S}(v)
\label{03}%
\end{equation}
and the proposition is proved.

Also, from (Eq.(\ref{A6})) we know that $\boldsymbol{\partial}_{u}%
\lrcorner\mathbf{P}_{u}~(\mathring{v}_{\parallel})=0$. So,
\begin{align}
\mathfrak{\mathring{d}}v  &  =\mathfrak{\mathring{d}(}\mathbf{P}%
(v))=\mathfrak{\mathring{d}}\mathbf{P~}(v)+\mathbf{P}(\mathfrak{\mathring{d}%
}v)\nonumber\\
&  =\mathfrak{\partial}_{u}\mathfrak{\wedge}\mathbf{P}_{u}\mathbf{~}%
(v)+\mathfrak{\partial}_{u}\mathfrak{\lrcorner}\mathbf{P}_{u}\mathbf{~}%
(v)+\boldsymbol{\partial}v\nonumber\\
&  =\mathfrak{\partial}_{u}\mathfrak{\wedge}\mathbf{P}_{u}\mathbf{~}%
(v)+\boldsymbol{\partial}\wedge v+\boldsymbol{\partial}\lrcorner v\nonumber\\
&  =\mathcal{S}(v)+\boldsymbol{\partial}\wedge v+\boldsymbol{\partial
}\lrcorner v \label{ss4a}%
\end{align}
and thus we see that%
\begin{equation}
\mathfrak{\mathring{d}\lrcorner}v=\boldsymbol{\partial}\lrcorner v \label{ss5}%
\end{equation}
We then can write
\begin{equation}
\mathfrak{\mathring{d}}v=\boldsymbol{\partial}v+\mathcal{S}(v). \label{ss6}%
\end{equation}

\subsection{$\mathfrak{\mathring{d}}\mathcal{C}=\boldsymbol{\partial
}\mathcal{C}+\mathbf{S}(\mathcal{C})$}

We can generalize Eq.(\ref{ss6}), i.e., we have the

\begin{proposition}
\label{ptoCC}For any\ $\mathcal{C}\in\sec$ $\mathcal{C}\ell(M,\mathtt{g})$ we
have
\begin{gather}
\mathfrak{\mathring{d}}\mathcal{C}=\boldsymbol{\partial}\mathcal{C}%
+\mathbf{S}(\mathcal{C}),\nonumber\\
\mathfrak{\mathring{d}}\wedge\mathcal{C}=\boldsymbol{\partial}\wedge
\mathcal{C}+\mathbf{S}(\mathcal{C}),~~~\mathfrak{\mathring{d}}\lrcorner
\mathcal{C}=\boldsymbol{\partial}\lrcorner\mathcal{C}. \label{g1}%
\end{gather}

\end{proposition}

\begin{proof}
\textbf{(i)} From the fact that for any $\mathcal{A},\mathcal{B}\in\sec$
$\mathcal{C}\ell(\mathring{M},\mathtt{\mathring{g}})$ it is $\mathbf{P}%
(\mathcal{A}\wedge\mathcal{B})=\mathbf{P}(\mathcal{A})\wedge\mathbf{P}%
(\mathcal{B})$ we have differentiating with respect to $u\in\sec%
{\textstyle\bigwedge\nolimits^{1}}
T^{\ast}M\hookrightarrow\sec$ $\mathcal{C}\ell(M,\mathtt{g})$
\begin{equation}
\mathbf{P}_{u}(\mathcal{A}\wedge\mathcal{B})=\mathbf{P}_{u}(\mathcal{A}%
)\wedge\mathcal{B}+\mathcal{A}\wedge\mathbf{P}_{u}(\mathcal{B}) \label{g2}%
\end{equation}
and of course
\begin{gather}
\mathbf{P}_{u}(\mathcal{A}_{\perp}\wedge\mathcal{B}_{\parallel})=\mathbf{P}%
_{u}(\mathcal{A}_{\perp})\wedge\mathcal{B}_{\parallel},~~~~\mathbf{P}%
_{u}(\mathcal{A}_{\perp}\wedge\mathcal{B}_{\perp})=0,\nonumber\\
\mathbf{P}_{u}(\mathcal{A}_{\parallel}\wedge\mathcal{B}_{\parallel
})=\mathbf{P}_{u}(\mathcal{A}_{\parallel})\wedge\mathcal{B}_{\parallel
}+\mathcal{A}_{\parallel}\wedge\mathbf{P}_{u}(\mathcal{B}_{\parallel}).
\label{g3}%
\end{gather}
(\textbf{ii}) For $\mathcal{C}\in\sec$ $\mathcal{C}\ell(M,\mathtt{g})$ it is
\ $\mathcal{C}=\mathbf{P}(\mathcal{C})$ and we have using Eq.(\ref{S1})%
\begin{align}
\mathfrak{\mathring{d}}\mathcal{C}  &  =\mathfrak{\mathring{d}}\mathbf{P~}%
(\mathcal{C})-\mathbf{P}(\mathfrak{\mathring{d}}\mathcal{C})\nonumber\\
&  =\mathfrak{\mathring{d}\wedge}\mathbf{P~}(\mathcal{C})+\mathfrak{\mathring
{d}\lrcorner}\mathbf{P~}(\mathcal{C})+\boldsymbol{\partial}\mathcal{C}
\label{g3a}%
\end{align}
(\textbf{iii}) Now, we can verify recalling that $\mathbf{S(}\mathcal{C}%
)=\mathbf{S(}\mathcal{C}_{\parallel}+\mathcal{C}_{\perp})=\mathbf{S(}%
\mathcal{C}_{\parallel})+\mathbf{S(}\mathcal{C}_{\perp})$ and following steps
analogous to the ones used in the proof of Proposition \ref{proS} that
\begin{equation}
\mathbf{S(}\mathcal{C}_{\parallel})=\mathbf{S(P(}\mathcal{C}_{\parallel
}))=\mathfrak{\mathring{d}\wedge}\mathbf{P~(}\mathcal{C}\mathbf{_{\parallel}%
)}~~~~\mathbf{S(}\mathcal{C}_{\perp})=\mathbf{P(S(}\mathcal{C}_{\perp
}))=\mathfrak{\mathring{d}\lrcorner}\mathbf{P~(}\mathcal{C}\mathbf{_{\parallel
}).} \label{G4}%
\end{equation}
(\textbf{iv}) Using Eq.(\ref{G4}) in Eq.(\ref{g3}) we have
\[
\mathfrak{\mathring{d}\wedge}\mathcal{C}+\mathfrak{\mathring{d}\lrcorner
}\mathcal{C}=\mathfrak{\mathring{d}\wedge}\mathbf{P~}(\mathcal{C}%
)+\mathfrak{\mathring{d}\lrcorner}\mathbf{P~}(\mathcal{C}%
)+\boldsymbol{\partial}\mathfrak{\wedge}\mathcal{C}+\boldsymbol{\partial
}\mathfrak{\lrcorner}\mathcal{C}%
\]
or
\begin{align}
\mathfrak{\mathring{d}\wedge}\mathcal{C}+\mathfrak{\mathring{d}\lrcorner
}\mathcal{C}  &  =\mathbf{S(}\mathcal{C})+\mathbf{S(}\mathcal{C}_{\perp
})+\boldsymbol{\partial}\mathfrak{\wedge}\mathcal{C}+\boldsymbol{\partial
}\mathfrak{\lrcorner}\mathcal{C}\nonumber\\
&  =\mathbf{S(}\mathcal{C})+\boldsymbol{\partial}\mathfrak{\wedge}%
\mathcal{C}+\boldsymbol{\partial}\mathfrak{\lrcorner}\mathcal{C}, \label{G5}%
\end{align}
which provides the proof of the proposition.
\end{proof}

\begin{proposition}
For any $\mathcal{C}\in\sec$ $\mathcal{C}\ell(M,\mathtt{g})$ we have:%
\begin{equation}
\boldsymbol{\partial}\mathcal{C}=\mathbf{P(}\mathfrak{\mathring{d}%
}\mathcal{C)} \label{DD1}%
\end{equation}

\end{proposition}

\begin{proof}
From Eq.(\ref{S1}) when\ $\mathcal{C}\in\sec$ $\mathcal{C}\ell(M,\mathtt{g})$
it is%
\[
\mathbf{P(S}(\mathcal{C)})=0.
\]
Then, applying $\mathbf{P}$ to both members of the first line of Eq.(\ref{g1})
we have
\begin{equation}
\mathbf{P(}\mathfrak{\mathring{d}}\mathcal{C)=}\mathbf{(}\mathfrak{\mathring
{d}}\mathcal{C)}_{\parallel}=\mathbf{P(}\partial\mathcal{C)}+\mathbf{P}%
^{2}\mathbf{(S}(\mathcal{C)})=\mathbf{P(}\boldsymbol{\partial}\mathcal{C)}%
=\boldsymbol{\partial}\mathcal{C} \label{ufa}%
\end{equation}
and the proposition is proved.
\end{proof}

\section{Curvature Biform $\mathfrak{R}(u\wedge v)$ Expressed in Terms of the
Shape Operator}

\subsection{Equivalent Expressions for $\mathfrak{R}(u\wedge v)$}

In this section we suppose that the structure $(M,\boldsymbol{g},D)$ \ is such
that is $M$ a submanifold of $\mathring{M}\simeq\mathbb{R}^{n}$ and $D$ the
Levi-Civita connection of $\boldsymbol{g=i}^{\ast}\boldsymbol{\mathring{g}}$.
We obtained in Section 1 a formula (Eq.(\ref{rie1})) for the curvature biform
$\mathfrak{R}(u\wedge v)$ of a general Riemann-Cartan connection. Of course,
taking into account the fact that $\mathfrak{R}$ is an intrinsic object, the
evaluation of $\mathfrak{R}(u\wedge v)$ does not depend on the coordinate
chart and basis for vector and form fields used for its calculation. In what
follows we take advantage of this fact choosing the basis $\{\frac{\partial
}{\partial\boldsymbol{x}^{\mathbf{i}}},d\boldsymbol{x}^{\mathbf{i}}\}$ as
introduced above for which $\mathring{\omega}(u)=0$. Thus, we have, recalling
Eq.(\ref{rie1}) and Eq.(\ref{curvature}) that%

\begin{align}
\mathfrak{R}(u\wedge v)  &  =D_{\boldsymbol{u}}~\omega(v)-D_{\boldsymbol{v}%
}~\omega(u)+\omega(u)\times\omega(v)-\boldsymbol{\omega}_{[\boldsymbol{u}%
,\boldsymbol{v}]}\nonumber\\
&  =\mathring{D}_{u}~\omega(v)-\mathring{D}_{\boldsymbol{v}}~\omega
(u)+\omega(u)\times\omega(v)\nonumber\\
&  -\omega(v)\times\omega(u)-\omega(u)\times\omega(v)-\boldsymbol{\omega
}_{[\boldsymbol{u},\boldsymbol{v}]}\nonumber\\
&  =\mathring{D}_{\boldsymbol{u}}~\omega(v)-\mathring{D}_{\boldsymbol{v}%
}~\omega(u)+\omega(u)\times\omega(v)-\boldsymbol{\omega}_{[\boldsymbol{u}%
,\boldsymbol{v}]}. \label{bi1}%
\end{align}
On the other hand since in the gauge where $\mathring{\omega}(u)=0$ we have
that $\omega(u)=\mathcal{S}(u)$ and thus we can also write%
\begin{equation}
\mathfrak{R}(u\wedge v)=\mathring{D}_{\boldsymbol{u}}~\omega(v)-\mathring
{D}_{\boldsymbol{v}}~\omega(u)+\mathcal{S}(u)\times\mathcal{S}(v)-\mathcal{S(}%
[u,v]). \label{BI2}%
\end{equation}
Now, putting $\mathbf{x}_{\mathbf{i}}=\partial/\partial\boldsymbol{x}%
^{\mathbf{i}}$ we have%

\begin{align}
&  \mathring{D}_{\boldsymbol{u}}~\omega(v)-\mathring{D}_{\boldsymbol{v}%
}~\omega(u)\nonumber\\
&  =-u^{\mathbf{i}}v^{\mathbf{j}}\{\mathring{D}_{\mathbf{x}_{\mathbf{i}}%
}\mathcal{S}(\vartheta_{\mathbf{j}})-\mathring{D}_{\mathbf{x}_{\mathbf{J}}%
}\mathcal{S}(\vartheta_{\mathbf{i}})\}\nonumber\\
&  =-u^{\mathbf{i}}v^{\mathbf{j}}\{\mathring{D}_{\mathbf{x}_{\mathbf{i}}%
}(\mathring{D}_{\mathbf{x}_{\mathbf{j}}}I_{m}~I_{m}^{-1})-\mathring
{D}_{\mathbf{x}_{\mathbf{j}}}(\mathring{D}_{\mathbf{x}_{\mathbf{i}}}%
I_{m}~I_{m}^{-1})\}\nonumber\\
&  =-u^{\mathbf{i}}v^{\mathbf{j}}\{\mathring{D}_{\mathbf{x}_{\mathbf{i}}%
}\mathring{D}_{\mathbf{x}_{\mathbf{j}}}I_{m})~I_{m}^{-1})+(\mathring
{D}_{\mathbf{x}_{\mathbf{j}}}I_{m})~(\mathring{D}_{\mathbf{x}_{\mathbf{i}}%
}I_{m}^{-1})\nonumber\\
&  -(\mathring{D}_{\mathbf{x}_{\mathbf{j}}}\mathring{D}_{\mathbf{x}%
_{\mathbf{i}}}I_{m})~I_{m}^{-1})-(\mathring{D}_{\mathbf{x}_{\mathbf{i}}}%
I_{m})~(\mathring{D}_{\mathbf{x}_{\mathbf{j}}}I_{m}^{-1})\}\nonumber\\
&  =-u^{\mathbf{i}}v^{\mathbf{j}}\{(\mathring{D}_{[\mathbf{x}_{\mathbf{i}%
},\mathbf{x}_{\mathbf{j}}]}I_{m}~I_{m}^{-1})-(\mathring{D}_{\mathbf{x}%
_{\mathbf{j}}}I_{m})~(\mathring{D}_{\mathbf{x}_{\mathbf{i}}}I_{m}%
^{-1})-(\mathring{D}_{\mathbf{x}_{\mathbf{i}}}I_{m})~(\mathring{D}%
_{\mathbf{x}_{\mathbf{j}}}I_{m}^{-1})\}\nonumber\\
&  =-u^{\mathbf{i}}v^{\mathbf{j}}\{-(\mathring{D}_{\mathbf{x}_{\mathbf{j}}%
}I_{m})~(\mathring{D}_{\mathbf{x}_{\mathbf{i}}}I_{m}^{-1})-(\mathring
{D}_{\mathbf{x}_{\mathbf{i}}}I_{m})~(\mathring{D}_{\mathbf{x}_{\mathbf{j}}%
}I_{m}^{-1})\}\nonumber\\
&  =-u^{\mathbf{i}}v^{\mathbf{j}}\{((\mathring{D}_{\mathbf{x}_{\mathbf{i}}%
}I_{m})I_{m}^{-1})~((\mathring{D}_{\mathbf{x}_{j}}I_{m}^{-1})I_{m}%
)-((\mathring{D}_{\mathbf{x}_{\mathbf{i}}}I_{m})I_{m}^{-1})~((\mathring
{D}_{\mathbf{x}_{j}}I_{m}^{-1})I_{m})\}\nonumber\\
&  =-u^{\mathbf{i}}v^{\mathbf{j}}\{((\mathring{D}_{\mathbf{x}_{\mathbf{i}}%
}I_{m})I_{m}^{-1})~((\mathring{D}_{\mathbf{x}_{j}}I_{m})I_{m}^{-1}%
)-((\mathring{D}_{\mathbf{x}_{\mathbf{j}}}I_{m})I_{m}^{-1})~((\mathring
{D}_{\mathbf{x}_{\mathbf{i}}}I_{m})I_{m}^{-1})\}\nonumber\\
&  =-u^{\mathbf{i}}v^{\mathbf{j}}\mathcal{S}(\vartheta_{\mathbf{i}%
})\mathcal{S}(\vartheta_{\mathbf{j}})+u^{\mathbf{i}}v^{\mathbf{j}}%
\mathcal{S}(\vartheta_{\mathbf{j}})\mathcal{S}(\vartheta_{\mathbf{i}%
})\nonumber\\
&  =-\mathcal{S}(u)\mathcal{S}(v)+\mathcal{S}(v)\mathcal{S}(u)=-2\mathcal{S}%
(u)\times\mathcal{S}(v). \label{BL3}%
\end{align}

Thus, we get
\begin{equation}
\mathfrak{R}(u\wedge v)=-\mathcal{S}(u)\times\mathcal{S}(v)-\mathcal{S}%
([u,v]). \label{bi4}%
\end{equation}
Now take into account that since $\mathfrak{R}(u\wedge v)\in\sec%
{\textstyle\bigwedge\nolimits^{2}}
T^{\ast}M\hookrightarrow\sec$ $\mathcal{C}\ell(M,\mathtt{g})$ we must have, of
course, $-\mathcal{S}(u)\times\mathcal{S}(v)-\mathcal{S}([u,v])=\mathbf{P}%
(-\mathcal{S}(u)\times\mathcal{S}(v)-\mathcal{S}([u,v]))$ and since
Eq.(\ref{P9}) tell us that $\mathbf{P}(\mathcal{S}([u,v])=0$ we have\ the nice
formula\footnote{Note that in \cite{hs1984,sobczyk} the second member of
Eq.(\ref{bi5}) is the negative of what we found. Our result agrees with the
one in \cite{dl}.}%
\begin{equation}
\mathfrak{R}(u\wedge v)=-\mathbf{P}(\mathcal{S}(u)\times\mathcal{S}(v))
\label{bi5}%
\end{equation}
which express the curvature biform in terms of the shape biform.

\subsection{$\mathbf{S}^{2}(v)=-\boldsymbol{\partial}\wedge
\boldsymbol{\partial}~(v)$}

In this subsection we want to show the

\begin{proposition}
\label{shapericci} Let $v\in\sec%
{\textstyle\bigwedge\nolimits^{1}}
T^{\ast}M\hookrightarrow\sec$ $\mathcal{C}\ell(M,\mathtt{g})$. Then,
\end{proposition}

\begin{equation}
\mathbf{S}^{2}(v)=-\boldsymbol{\partial}\wedge\boldsymbol{\partial}~(v).
\label{bi6}%
\end{equation}

Eq.(\ref{bi6}) tell us that the square of the shape operator applied to a
$1$-form field $v$ is equal to the Ricci operator applied to $v$. We will
comment more on the significance of this result in the conclusions.

Now, to prove the Proposition \ref{shapericci} we need the following lemmas

\begin{lemma}
Let $\mathcal{C\in}\sec$ $\mathcal{C}\ell(\mathring{M},\mathtt{g})$ and
$v\in\sec%
{\textstyle\bigwedge\nolimits^{1}}
T^{\ast}M\hookrightarrow\sec$ $\mathcal{C}\ell(M,\mathtt{g})$.Then
\end{lemma}

\begin{equation}
\mathbf{P}_{v}(\mathcal{C})=\mathbf{P}(\mathcal{C})\times\mathcal{S(}%
v\mathcal{)}-\mathbf{P}(\mathcal{C}\times\mathcal{S(}v\mathcal{))}. \label{l1}%
\end{equation}

\begin{proof}
Indeed,%
\begin{align}
\mathbf{P}_{v}(\mathcal{C})  &  =\mathring{D}_{\boldsymbol{v}}(\mathbf{P}%
(\mathcal{C}))-\mathbf{P}(\mathring{D}_{\boldsymbol{v}}\mathcal{C})\nonumber\\
&  =D_{\boldsymbol{v}}(\mathbf{P}(\mathcal{C}))-\mathcal{S}(v)\times
\mathbf{P}(\mathcal{C})-\mathbf{P}(D_{\boldsymbol{v}}\mathcal{C}%
-\mathcal{S}(v)\times\mathcal{C})\nonumber\\
&  =D_{\boldsymbol{v}}\mathcal{C}-\mathcal{S}(v)\times\mathbf{P}%
(\mathcal{C})-D_{\boldsymbol{v}}\mathcal{C}+\mathbf{P}(\mathcal{S}%
(v)\times\mathcal{C})\nonumber\\
&  =\mathbf{P}(\mathcal{C})\times\mathcal{S(}v\mathcal{)}-\mathbf{P}%
(\mathcal{C}\times\mathcal{S(}v\mathcal{))} \label{l4}%
\end{align}
which proves the lemma.
\end{proof}

\begin{lemma}
Let $\mathcal{C\in}\sec$ $\mathcal{C}\ell(M,\mathtt{g})$ and $v\in\sec%
{\textstyle\bigwedge\nolimits^{1}}
T^{\ast}M\hookrightarrow\sec$ $\mathcal{C}\ell(M,\mathtt{g})$.Then%
\begin{equation}
D_{\boldsymbol{v}}\mathcal{C}=\mathring{D}_{\boldsymbol{v}}\mathcal{C-}%
\mathbf{P}_{v}(\mathcal{C}). \label{l5}%
\end{equation}

\end{lemma}

\begin{proof}
Follows from the first line in Eq.(\ref{l4}).
\end{proof}

\begin{lemma}
Let $\mathcal{C\in}\sec$ $\mathcal{C}\ell(M,\mathtt{g})$ and $u,v\in\sec%
{\textstyle\bigwedge\nolimits^{1}}
T^{\ast}M$.$\hookrightarrow\sec$ $\mathcal{C}\ell(M,\mathtt{g})$.Then%
\begin{equation}
D_{\boldsymbol{u}}D_{\boldsymbol{v}}\mathcal{C}=\mathbf{P}(\mathring
{D}_{\boldsymbol{u}}\mathring{D}_{\boldsymbol{v}}\mathcal{C})+\mathbf{P}%
_{u}\mathbf{P}_{v}(\mathcal{C}). \label{ll1}%
\end{equation}

\end{lemma}

\begin{proof}
Using Eq.(\ref{l5}) we have
\begin{align}
D_{\boldsymbol{u}}(D_{\boldsymbol{v}}\mathcal{C)}  &  =D_{\boldsymbol{u}%
}(\mathring{D}_{\boldsymbol{v}}\mathcal{C-}\mathbf{P}_{v}(\mathcal{C}%
))\nonumber\\
&  =\mathring{D}_{\boldsymbol{u}}\mathring{D}_{\boldsymbol{v}}\mathcal{C-}%
\mathbf{P}_{u}(\mathring{D}_{\boldsymbol{v}}\mathcal{C})-\mathring
{D}_{\boldsymbol{u}}(\mathbf{P}_{v}(\mathcal{C}))+\mathbf{P}_{u}\mathbf{P}%
_{v}(\mathcal{C}). \label{ll1a}%
\end{align}

On the other hand we have%
\begin{align}
\mathbf{P}(\mathring{D}_{\boldsymbol{u}}\mathring{D}_{\boldsymbol{v}%
}\mathcal{C})  &  =-\mathbf{P}_{u}(\mathring{D}_{\boldsymbol{v}}%
\mathcal{C})+\mathring{D}_{\boldsymbol{u}}(\mathbf{P}(\mathring{D}%
_{\boldsymbol{v}}\mathcal{C}))\nonumber\\
&  =-\mathbf{P}_{u}(\mathring{D}_{\boldsymbol{v}}\mathcal{C})+\mathring
{D}_{\boldsymbol{u}}(D_{\boldsymbol{v}}\mathcal{C})\nonumber\\
&  =-\mathbf{P}_{u}(\mathring{D}_{\boldsymbol{v}}\mathcal{C})+\mathring
{D}_{\boldsymbol{u}}(\mathring{D}_{\boldsymbol{v}}\mathcal{C-}\mathbf{P}%
_{v}(\mathcal{C}))\nonumber\\
&  =\mathring{D}_{\boldsymbol{u}}\mathring{D}_{\boldsymbol{v}}\mathcal{C-}%
\mathbf{P}_{u}(\mathring{D}_{\boldsymbol{v}}\mathcal{C})-\mathring
{D}_{\boldsymbol{u}}(\mathbf{P}_{v}(\mathcal{C})). \label{ll2}%
\end{align}
Putting Eq.(\ref{ll2}) in Eq.(\ref{ll1}) gives the desired result.
\end{proof}

\begin{lemma}
Let $\mathcal{C\in}\sec$ $\mathcal{C}\ell(M,\mathtt{g})$ and $u,v\in\sec%
{\textstyle\bigwedge\nolimits^{1}}
T^{\ast}M\hookrightarrow\sec$ $\mathcal{C}\ell(M,\mathtt{g})$.Then
\end{lemma}

\begin{equation}
\mathfrak{R}(u\wedge v)\times\mathcal{C}=[\mathbf{P}_{u},\mathbf{P}%
_{v}]\mathcal{C}. \label{II3a}%
\end{equation}

\begin{proof}
Using Eq.(\ref{ll1}) we have%
\begin{align}
\lbrack D_{\boldsymbol{u}},D_{\boldsymbol{v}}]\mathcal{C}  &  =\mathbf{P}%
([\mathring{D}_{\boldsymbol{u}},\mathring{D}_{\boldsymbol{v}}]\mathcal{C)}%
+[\mathbf{P}_{u},\mathbf{P}_{v}]\mathcal{C}\nonumber\\
&  =\mathbf{P}(\mathring{D}_{[\boldsymbol{u,v]}}\mathcal{C)}+[\mathbf{P}%
_{u},\mathbf{P}_{v}]\mathcal{C}\nonumber\\
&  =D_{[\boldsymbol{u,v]}}\mathcal{C}+[\mathbf{P}_{u},\mathbf{P}%
_{v}]\mathcal{C} \label{ll4}%
\end{align}
Thus we get that%
\begin{equation}
([D_{\boldsymbol{u}},D_{\boldsymbol{v}}]-D_{[\boldsymbol{u,v]}})\mathcal{C}%
=[\mathbf{P}_{u},\mathbf{P}_{v}]\mathcal{C} \label{ll5}%
\end{equation}
Recalling now Eq.(\ref{exercise}) we have%
\[
\mathfrak{R}(u\wedge v)\times\mathcal{C}=[\mathbf{P}_{u},\mathbf{P}%
_{v}]\mathcal{C}%
\]
and the lemma is proved
\end{proof}

\begin{lemma}
Let $\mathcal{C\in}\sec\mathcal{C}\ell(M,\mathtt{g})$ and $u,v\in\sec%
{\textstyle\bigwedge\nolimits^{1}}
T^{\ast}M$.$\hookrightarrow\sec\mathcal{C}\ell(M,\mathtt{g})$. Then,%

\begin{equation}
\mathfrak{R}(u\wedge v)\times\mathcal{C}=-\mathbf{P}(\mathcal{S}%
(u)\times\mathcal{S}(v))\times\mathcal{C} \label{ll6}%
\end{equation}

\end{lemma}

\begin{proof}
This follows directly from Eq.(\ref{bi5})\emph{.}
\end{proof}

\begin{remark}
We shall now evaluate directly the first member of \emph{Eq.(\ref{ll8}) to get
Eq.(\ref{oi}) which when compared with Eq.(\ref{bi5}) will furnish identities
given by Eq. (\ref{ll9}).}%
\begin{equation}
([D_{\boldsymbol{u}},D_{\boldsymbol{v}}]-D_{[\boldsymbol{u,v]}})\mathcal{C}%
=\mathfrak{R}(u\wedge v)\times\mathcal{C}. \label{ll8}%
\end{equation}
Given the linearity of $\mathfrak{R}(u\wedge v)$ we calculate the first member
of \emph{Eq.(\ref{ll8})} for the case $\boldsymbol{u}=\mathbf{x}_{\mathbf{i}}%
$, $\boldsymbol{v}$ $=\mathbf{x}_{\mathbf{j}}$. Taking into account that
$D_{\boldsymbol{u}}\mathcal{C}=\mathring{D}_{\boldsymbol{u}}\mathcal{C+S}%
(u)\times C$ we get with calculations analogous to the ones in
\emph{Eq.(\ref{BL3})} that
\end{remark}

\begin{equation}
\lbrack D_{\mathbf{x}_{\mathbf{i}}},D_{\mathbf{x}_{\mathbf{j}}}]\mathcal{C}%
=-\mathcal{S}(\vartheta_{\mathbf{i}})\times\mathcal{S}(\vartheta_{\mathbf{j}%
})\times\mathcal{C}%
\end{equation}
and taking into account that it is $[\mathbf{x}_{\mathbf{i}},\mathbf{x}%
_{\mathbf{j}}]=0$ we can write the last equation as%
\begin{equation}
([D_{\mathbf{x}_{\mathbf{i}}},D_{\mathbf{x}_{\mathbf{j}}}]-D_{[\mathbf{x}%
_{\mathbf{i}},\mathbf{x}_{\mathbf{j}}]})\mathcal{C}=\mathfrak{R}%
(\vartheta_{\mathbf{i}}\wedge\vartheta_{\mathbf{j}})\times\mathcal{C=}%
-\mathcal{S}(\vartheta_{\mathbf{i}})\times\mathcal{S}(\vartheta_{\mathbf{j}%
})\times\mathcal{C}%
\end{equation}
and so it follows that
\begin{equation}
\mathfrak{R}(u\wedge v)\times\mathcal{C=}-\mathcal{S}(u)\times\mathcal{S}%
(v)\times\mathcal{C}. \label{oi}%
\end{equation}
Of course, we must have $\mathbf{P}(\mathcal{S}(u)\times\mathcal{S}%
(v)\times\mathcal{C})$ $\in\sec$ $\mathcal{C}\ell(M,\mathtt{g})$. Since
$\mathcal{S}(u)\times\mathcal{S}(v)=\mathbf{P}(\mathcal{S}(u)\times
\mathcal{S}(v))+\mathbf{P}_{\perp}(\mathcal{S}(u)\times\mathcal{S}(v))$ and we
already know that $\mathfrak{R}(u\wedge v)=-\mathbf{P}(\mathcal{S}%
(u)\times\mathcal{S}(v))$\ it follows that $\mathbf{P}_{\perp}(\mathcal{S}%
(u)\times\mathcal{S}(v))=0$ and moreover we get that
\begin{equation}
\mathbf{P}(\mathcal{S}(u)\times\mathcal{S}(v)\times\mathcal{C})=\mathbf{P}%
(\mathcal{S}(u)\times\mathcal{S}(v))\times\mathcal{C}=\mathbf{P}%
(\mathcal{S}(u)\times\mathcal{S}(v))\times\mathbf{P}(\mathcal{C}). \label{ll9}%
\end{equation}

\begin{lemma}
Let $\mathcal{C\in}\sec$ $\mathcal{C}\ell(M,\mathtt{g})$ and $u,v\in\sec%
{\textstyle\bigwedge\nolimits^{1}}
T^{\ast}M\hookrightarrow\mathcal{\in}\sec$ $\mathcal{C}\ell(M,\mathtt{g})$%
\begin{equation}
\mathfrak{R}(u\wedge v)=\mathbf{P}_{v}(\mathcal{S}(u)) \label{II10}%
\end{equation}
\ 
\end{lemma}

\begin{proof}
Taking $\mathcal{C}=\mathcal{S(}u\mathcal{)}$ in Eq.(\ref{l4})\emph{ }and
recalling\emph{ }Eq.(\ref{P9})\emph{ }$\mathbf{P(}\mathcal{S(}u))=0.$we get%
\begin{equation}
\mathbf{P}_{v}(\mathcal{S(}u))=-\mathbf{P}(\mathcal{S(}u)\times\mathcal{S(}%
v\mathcal{))} \label{ll11}%
\end{equation}
which proves the lemma.
\end{proof}

\begin{remark}
From \emph{Eq.(\ref{ll11})}\ we immediately have
\begin{equation}
\mathbf{P}_{u}(\mathcal{S(}v))=-\mathbf{P}(\mathcal{S(}v)\times\mathcal{S(}%
u\mathcal{))}=\mathbf{P}(\mathcal{S(}u)\times\mathcal{S(}v\mathcal{))}%
=-\mathbf{P}_{v}(\mathcal{S(}u))=-\mathbf{P}_{v}(\mathbf{S}\mathcal{(}u))
\label{ll12}%
\end{equation}
where the last term follows from the fact that $\mathbf{S}\mathcal{(}%
u)=\mathcal{S(}u)$.
\end{remark}

\begin{proof}
(of Proposition\emph{ }\ref{shapericci}) We know that $\mathcal{R(}%
v\mathcal{)}=\partial_{u}\mathfrak{R}(u\wedge v)$. Thus using Eq.(\ref{ll12})
and recalling Eq.(\ref{S1}) we can write
\begin{align}
\mathcal{R(}v\mathcal{)}  &  =\partial_{u}\mathbf{P}_{v}(\mathcal{S(}%
u))=-\partial_{u}\mathbf{P}_{u}(\mathcal{S(}v))\label{ll12a}\\
&  =-\mathfrak{\mathring{d}}\mathbf{P~}(\mathbf{S}\mathcal{(}v))=-\mathbf{S}%
(\mathbf{S}\mathcal{(}v))=-\mathbf{S}^{2}(v).\nonumber
\end{align}
Since we already showed that $\mathcal{R(}v\mathcal{)}=\boldsymbol{\partial
}\wedge\boldsymbol{\partial~}(v)$ we get
\[
\boldsymbol{\partial}\wedge\boldsymbol{\partial~}(v)=-\mathbf{S}^{2}(v)
\]
and the proposition is proved.
\end{proof}

\begin{remark}
Note that whereas $\mathbf{S}\mathcal{(}v)$ is a section of $\mathcal{C\ell
(}\mathring{M},\mathtt{\mathring{g}}\mathcal{)}$, $\mathbf{S}^{2}%
\mathcal{(}v)\in\sec%
{\textstyle\bigwedge\nolimits^{1}}
T^{\ast}M\hookrightarrow\sec\mathcal{C}\ell(M,\mathtt{g})$
\end{remark}

\begin{proposition}
\label{procommp}Let $u,v,w\in\sec%
{\textstyle\bigwedge\nolimits^{1}}
T^{\ast}M\hookrightarrow\sec\mathcal{C}\ell(M,\mathtt{g})$. Then,%

\begin{equation}
\mathfrak{R}(u\wedge v)=\frac{1}{2}\partial_{w}\wedge\lbrack\mathbf{P}%
_{v},\mathbf{P}_{u}](w). \label{II16K}%
\end{equation}

\end{proposition}

\begin{proof}
From Eq.(\ref{II3a}) with $\mathcal{C}=w\in\sec%
{\textstyle\bigwedge\nolimits^{1}}
T^{\ast}M\hookrightarrow\sec\mathcal{C}\ell(M,\mathtt{g})$ we have%

\begin{equation}
\mathfrak{R}(u\wedge v)\times w=[\mathbf{P}_{u},\mathbf{P}_{v}](w).
\label{II16A}%
\end{equation}
Now, the first member of Eq.(\ref{II16A}) is%
\[
\mathfrak{R}(u\wedge v)\times w=-w\lrcorner\mathfrak{R}(u\wedge v)
\]
Now, writing $\mathfrak{R}(u\wedge v)=\frac{1}{2}u^{\mathbf{i}}%
vR_{\mathbf{ij\cdot\cdot}}^{\cdot\cdot\mathbf{kl}}\theta_{\mathbf{k}}%
\wedge\theta_{\mathbf{l}}$ we have
\begin{align}
\partial_{w}\wedge(\mathfrak{R}(u\wedge v)\times w)  &  =-\theta^{\mathbf{r}%
}\wedge(\theta_{\mathbf{r}}\lrcorner\mathfrak{R}(u\wedge v))\nonumber\\
-\frac{1}{2}u^{\mathbf{i}}v^{\mathbf{j}}\theta_{\mathbf{r}}\wedge
(\theta^{\mathbf{r}}\lrcorner R_{\mathbf{ij}}^{~~\mathbf{kl}}\theta
_{\mathbf{k}}\wedge\theta_{\mathbf{l}})  &  =-2\mathfrak{R}(u\wedge v).
\label{II16B}%
\end{align}
Taking into account Eq.(\ref{II16A}) and Eq.(\ref{II16B}) the proof follows.

We can also prove the proposition as follows: directly from Eq.(\ref{ll5}) we
can write%
\begin{equation}
\lbrack\mathbf{P}_{u},\mathbf{P}_{v}](w)=u^{\mathbf{k}}v^{\mathbf{l}%
}([D_{\boldsymbol{e}_{\mathbf{k}}},D_{\boldsymbol{e}_{\mathbf{l}}%
}]-D_{[\boldsymbol{e}_{\mathbf{k}}\boldsymbol{,e}_{\mathbf{l}}\boldsymbol{]}%
})w=u^{\mathbf{k}}v^{\mathbf{l}}w^{\mathbf{j}}R_{\cdot\mathbf{jkl}%
}^{\mathbf{i\cdot\cdot\cdot}}\theta_{\mathbf{i}}. \label{ll20}%
\end{equation}
Thus
\begin{align}
\frac{1}{2}\partial_{w}\wedge\lbrack\mathbf{P}_{u},\mathbf{P}_{v}](w)  &
=\frac{1}{2}\theta^{m}\frac{\partial}{\partial w^{m}}\wedge(u^{\mathbf{k}%
}v^{\mathbf{l}}w^{\mathbf{j}}R_{\cdot\mathbf{jkl}}^{\mathbf{i\cdot\cdot\cdot}%
}\theta_{\mathbf{i}})\nonumber\\
&  =\frac{1}{2}u^{\mathbf{k}}v^{\mathbf{l}}R_{\mathbf{imkl}}\theta
^{\mathbf{m}}\wedge\theta^{\mathbf{i}}\nonumber\\
&  =-u^{\mathbf{k}}v^{\mathbf{l}}\mathcal{R}_{\mathbf{kl}}=-\mathfrak{R}%
(u\wedge v) \label{ll21}%
\end{align}
and the proof is complete.
\end{proof}

\begin{proposition}
\label{prodif}Let $u,v,w\in\sec%
{\textstyle\bigwedge\nolimits^{1}}
T^{\ast}M\hookrightarrow\sec\mathcal{C}\ell(M,\mathtt{g})$. Then,%

\begin{equation}
\mathfrak{R}(u\wedge v)=\partial_{w}\wedge\mathbf{P}_{v}\mathbf{P}_{u}(w).
\label{np1}%
\end{equation}

\end{proposition}

\begin{proof}
Recall that we proved (Eq.(\ref{II10})) that \ $\mathfrak{R}(u\wedge
v)=\mathbf{P}_{v}(\mathcal{S}(u))$. Also Eq.(\ref{A4}) says that
$\mathbf{S}(u)=\mathcal{S}(u)=\partial_{w}\wedge\mathbf{P}_{w}(u)$ for any
$u,w\in\sec%
{\textstyle\bigwedge\nolimits^{1}}
T^{\ast}M\hookrightarrow\sec\mathcal{C}\ell(M,\mathtt{g})$. Now from
Eq.(\ref{ap4}) we have
\begin{equation}
\mathbf{P}_{w}(u)=\mathbf{P}_{u}(w)=u\cdot\mathfrak{\mathring{d}}%
\mathbf{P~(}w\mathbf{)}=\mathring{D}\mathbf{_{u}P}~(w)=\mathring
{D}\mathbf{_{u}(P(}w))-\mathbf{P}(\mathring{D}\mathbf{_{u}}w)=(\mathring
{D}\mathbf{_{u}}w)_{\perp} \label{np2}%
\end{equation}
which means that
\begin{equation}
\mathbf{P}_{w}(u)=(\mathbf{P}_{w}(u))_{\perp}. \label{np3}%
\end{equation}
Then, we have that
\begin{align*}
\mathfrak{R}(u\wedge v)  &  =\mathbf{P}_{v}(\mathcal{S}(u))=\mathbf{P}%
_{v}(\partial_{w}\wedge\mathbf{P}_{w}(u))\\
&  =\mathbf{P}_{v}(\partial_{w}\wedge\mathbf{P}_{u}(w))=\mathbf{P}%
_{v}((\partial_{w})_{\parallel}\wedge(\mathbf{P}_{u}(w))_{\perp})\\
&  \overset{\text{Eq}.(\ref{ap7})}{=}\partial_{w}\wedge\mathbf{P}%
_{v}\mathbf{P}_{u}(w))_{\perp}%
\end{align*}
and the proof is complete.
\end{proof}

\begin{remark}
Since $\mathfrak{R}(u\wedge v)=-\mathfrak{R}(v\wedge u)$, \emph{Eq.(\ref{np1}%
)}\ implies that $u,v,w\in\sec%
{\textstyle\bigwedge\nolimits^{1}}
T^{\ast}M\hookrightarrow\sec\mathcal{C}\ell(M,\mathtt{g})$%
\begin{equation}
\partial_{w}\wedge\mathbf{P}_{v}\mathbf{P}_{u}(w)=-\partial_{w}\wedge
\mathbf{P}_{u}\mathbf{P}_{v}(w), \label{np4}%
\end{equation}
thus exhibiting the consistency of \emph{Eq.(\ref{np1})} with
\emph{Eq.(\ref{II16K}).}
\end{remark}

\section{On Clifford's Little Hills}

One could think that the fact that $\boldsymbol{\partial}\wedge
\boldsymbol{\partial~}(v)=\mathcal{R}(v)=-\mathbf{S}^{2}(v)$ when applied to
General Relativity coupled with brane theory permits to give a mathematical
formalization to Clifford's intuition\footnote{Taking into account, of course,
that differently from Clifford's idea,\ instead of a space theory of matter,
we must talk about a spacetime theory of matter.} presented in \cite{clifford}%
, namely that:

\begin{quote}
(1) That small portions of space are in fact of a nature analogous to little
hills on a surface which is on the average flat; namely, that the ordinary
laws of geometry are not valid in them.

(2) That this property of being curved or distorted is continually being
passed on from one portion of space to another after the manner of a wave.

(3) That this variation of the curvature of space is what really happens in
that phenomenon which we call the motion of matter, whether ponderable or ethereal.

(4) That in the physical world nothing else takes place but this variation,
subject (possibly) to the law of continuity.
\end{quote}

Let us see how to proceed. Let\footnote{The symbol $\uparrow$ means that the
Lorentzian manifold $(M,\boldsymbol{g})$ is time orientable. Details in
\cite{rodcap2007}.} $(M,\boldsymbol{g},D,\tau_{g},\uparrow)$ be a model of a
gravitational field generated by an energy momentum tensor $T^{\mathbf{a}%
}:=T_{\mathbf{b}}^{\mathbf{a}}\theta^{\mathbf{a}}\otimes\theta^{\mathbf{b}}$
describing all matter of the universe according to General Relativity theory.
As well known Einstein equation can be written as%
\begin{equation}
\boldsymbol{\partial\wedge}\partial~\theta^{\mathbf{a}}=-\mathcal{T}%
^{\mathbf{a}}+\frac{1}{2}\mathcal{T}\theta^{\mathbf{a}}, \label{e1}%
\end{equation}
where $\mathcal{T}^{\mathbf{a}}:=T_{\mathbf{b}}^{\mathbf{a}}\theta
^{\mathbf{b}}$ and $\mathcal{T}:=T_{\mathbf{a}}^{\mathbf{a}}$, with
$T_{\mathbf{b}}^{\mathbf{a}}$. If we suppose that the structure
$(M,\boldsymbol{g})$ is a submanifold of $(\mathring{M}\simeq\mathbb{R}%
^{n},\boldsymbol{\mathring{g}})$ for $n$ large enough as discussed in the
beginning of Section 3 we can write Eq.(\ref{e1}) taking into account
Eq.(\ref{bi6}) as
\begin{equation}
\mathbf{S}^{2}(\theta^{\mathbf{a}})=\mathcal{T}^{\mathbf{a}}-\frac{1}%
{2}\mathcal{T}\theta^{\mathbf{a}}. \label{e2}%
\end{equation}
So, in a region where there is no matter $\mathbf{S}^{2}(\theta^{\mathbf{a}%
})=0$, despite the fact that $\mathbf{S}(\theta^{\mathbf{a}})=\mathcal{S}%
(\theta^{\mathbf{a}})$ may be \emph{non} null. So, \ a being living in the
hyperspace $\mathbb{R}^{n}$ and looking at our brane world will see the little
hills (i.e., \textquotedblleft matter\textquotedblright) \ are special shapes
in $M$, places where the $\mathbf{S}^{2}(\theta^{\mathbf{a}})\neq0$ which act
as sources for $\mathbf{P}(\mathfrak{\mathring{d}\lrcorner}\mathcal{S}%
(\theta^{\mathbf{a}}))$ since $\mathbf{P}(\mathfrak{\mathring{d}\lrcorner
}\mathcal{S}(\theta^{\mathbf{a}}))=-\mathbf{S}^{2}(\theta^{\mathbf{a}})$.

\begin{remark}
To properly appreciate the above argument one must take in mind that the shape
extensor\ depends for its definition on the metric $\boldsymbol{\mathring{g}}$
and the Levi-Civita $\mathring{D}$ connection\ of\ $\boldsymbol{\mathring{g}}$
used in $\mathring{M}$. So, a different choice of metric in $\mathring{M}$
will imply in Clifford's little hills to be represented by different shape
extensors. Despite this fact, it seems to us that shape is most appealing than
the curvature biform\footnote{Recall that $\mathfrak{R}$ is in general non
null even in vacuum.} or the Ricci $1$-form fields $\mathcal{R}^{\mathbf{a}%
}=\boldsymbol{\partial\wedge}\partial~\theta^{\mathbf{a}}$ as indicator of the
presence of matter as distortions in the world brane\ $M$. Indeed, inner
observers living in $M$ in general may not have enough skills and technology
to discover the topology \ of $M$ and\ so cannot know if their brane world is
a bended surface in the hyperspace \emph{(}i.e., $\mathring{M}$\emph{) }or
even if a open set\emph{ }$U\subset M$ is a part of an hyperplane or
not\emph{. }Moreover, those inner observers\ that have learned a little bit of
differential geometry know that they cannot say that their manifold is curved
based on the fact that the curvature biform is non null, for they know that
the curvature biform is a property of the connection \emph{(}parallelism
rule\emph{) that }they decide to use by convention in $M$ and not an intrinsic
property of $M$. They know that if they choose a different connection it may
happen that its curvature biform may be null and their connection (not their
manifold) may have torsion and even a non null nonmetricity
tensor\footnote{Details about these possibilities are discussed in
\cite{fr2010} where a theory of the gravitational field on a brane
diffeomorphic to $R^{4}$ is discussed.} So, with their knowledge of
differential geometry they infer that little hills \emph{(}as seems for beings
living in $\mathring{M}$\emph{)} can only be associated to the shape extensor
if they use Levi-Civita connection of $g$ in $M$.
\end{remark}

\section{A Maxwell Like Equation for a Brane World with a Killing Vector
Field}

When $(M,\boldsymbol{g})$ admits a Killing vector field $\boldsymbol{A}\in\sec
TM$ then it follows \cite{rod2010} that $\delta A=0$, where $A=\boldsymbol{g}%
(\boldsymbol{A,~})\in\sec%
{\textstyle\bigwedge\nolimits^{1}}
T^{\ast}M\hookrightarrow\sec\mathcal{C\ell(}M,\mathtt{g}\mathcal{)}$. In this
case we can show that the Ricci operator applied to $A$ is equal to the
covariant D'Alembertian operator applied to $A$, i.e.,
\begin{equation}
\boldsymbol{\partial\wedge\partial}~A=\boldsymbol{\partial\cdot\partial}~A
\label{E4}%
\end{equation}
Now, recalling Eq.(\ref{3}) that the square of the Dirac operator
$\boldsymbol{\partial}^{2}$ can be decomposed in two ways, i.e.,
\begin{equation}
\boldsymbol{\partial\wedge\partial}~A+\boldsymbol{\partial\cdot\partial
}~A=\boldsymbol{\partial}^{2}A=-d\delta A-\delta dA \label{e5}%
\end{equation}
we have writing $F=dA$ and taking into account that $\delta A=0$ that Einstein
equation can be written as
\begin{equation}
\delta F=2\mathbf{S}^{2}(A) \label{E6}%
\end{equation}
and since $dF=ddA=0$ we can write Einstein equation as:
\begin{equation}
\boldsymbol{\partial}F=-2\mathbf{S}^{2}(A). \label{e7}%
\end{equation}

Eq.(\ref{e7}) shows that in a Lorentzian brane $M$ of dim $4$ which contains a
Killing vector field $\boldsymbol{A}$, Einstein equation is encoded in an
\textquotedblleft electromagnetic like field\textquotedblright\ $F$ having as
source a current $J=2\mathbf{S}^{2}(A)\in\sec\mathcal{C\ell(}M,\mathtt{g}%
\mathcal{)}$.

\section{Conclusions}

In this paper, we gave a thoughtful presentation of the geometry of vector
manifolds using the Clifford bundle formalism, hopping to provide a useful
text for people (who know the Cartan theory of differential forms)
\footnote{This includes people, we think, interested in string and brane
theories and General Relativity.} and are interested in the differential
geometry of submanifolds $M$ (of dimension $m$ equipped with a
metric\footnote{$\boldsymbol{i}:M\rightarrow\mathring{M}$ is the inclusion
map.} $\boldsymbol{g=i}^{\ast}\boldsymbol{\mathring{g}}$ of signature $(p,q)$
and its Levi-Civita connection $D$) of a manifold $\mathring{M}\simeq
\mathbb{R}^{n}$ (of dimension $n$ and equipped with a metric
$\boldsymbol{\mathring{g}}$ of signature $(\mathring{p},\mathring{q})$ and its
Levi-Civita connection $\mathring{D}$). We proved in details several
equivalent expressions for the curvature biforms $\mathfrak{R}(u\wedge v)$ and
moreover proved that the Ricci operator $\boldsymbol{\partial}\wedge
\boldsymbol{\partial}$ when applied to a $1$-form field $v$ is such that
$\boldsymbol{\partial}\wedge\boldsymbol{\partial~}(v)=\mathcal{R}%
(v)=-\mathbf{S}^{2}(v)$ ($\mathcal{R}(v)=R_{\mathbf{b}}^{\mathbf{a}}%
\theta_{\mathbf{b}}$) is the negative of the square of the shape operator
$\mathbf{S}$.\ We showed in Section 5 that when this result is applied to
General Relativity it permits to give a mathematical realization to Clifford's
theory of matter. Moreover in Section 6 we show that in a Lorentzian brane
containing a Killing vector field Einstein equation can be encoded in a
Maxwell like equation whose source is a current given by $J=2\mathbf{S}%
^{2}(A).$

To end we observe that although some (but not all) of the results in this
paper appear in \cite{hs1984,h1986,sobczyk}, our methodology and many proofs
differs considerably. We use of the Clifford bundle of differential forms
$\mathcal{C\ell(}M,\mathtt{g}\mathcal{)}$) and give detailed and (we hope)
intelligible proofs of all formulas, clarifying some important issues,
presenting, e.g., the precise relation between the shape biform $S$ evaluate
at $v$ (a $1$-form field) and the connection extensor $\omega$ evaluated at
$v$ (Eq.(\ref{p21})). In particular, our approach also generalizes for a
general Riemann-Cartan connection the results in \cite{h1986} which are valid
only for the Levi-Civita connection $D$ of a Lorentzian metric of signature
$(1,3)$. Moreover our approach makes rigorous the results in \cite{h1986}
which are valid only for 4-dimensional Lorentzian spacetimes admiting spinor
structures\footnote{See \cite{geroch} for details.
\par
{}}, since in \cite{h1986} it is postulated that the frame bundle of $M$ has a
global section (there called a fiducial frame).

\appendix{}

\section{Some Identities Involving $\mathbf{P}$ and $\mathbf{P}_{u}$}

The projection operator $\mathbf{P}$ has been defined by Eq.(\ref{proj}) and
its covariant derivative $\mathbf{P}_{u}:=u\cdot\mathfrak{\mathring{d}%
}\mathbf{P}$ has been defined by Eq.(\ref{A1a}). Let \ $\mathcal{C}%
,\mathcal{D}\in\sec\mathcal{C}\ell(\mathring{M},\boldsymbol{\mathring{g}})$.
Since%
\begin{equation}
\mathbf{P}(\mathcal{C}\wedge\mathcal{D})=\mathbf{P}(\mathcal{C)}%
\wedge\mathbf{P}(\mathcal{D}), \label{ap1}%
\end{equation}
we have for any $u\in\sec%
{\textstyle\bigwedge\nolimits^{1}}
T^{\ast}M\hookrightarrow\sec\mathcal{C}\ell(M,\boldsymbol{g})$ that
\begin{equation}
\mathbf{P}_{u}(\mathcal{C}\wedge\mathcal{D})=\mathbf{P}_{u}(\mathcal{C)}%
\wedge\mathbf{P}(\mathcal{D})+\mathbf{P}(\mathcal{C)}\wedge\mathbf{P}%
_{u}(\mathcal{D}). \label{ap2}%
\end{equation}

From $\mathbf{P}^{2}(\mathcal{C)}=\mathbf{P}(\mathcal{C})$ we have that
\begin{equation}
\mathbf{P}_{u}\mathbf{P}(\mathcal{C})+\mathbf{PP}_{u}(\mathcal{C}%
)=\mathbf{P}_{u}(\mathcal{C}). \label{ap3}%
\end{equation}

Now, we easily verify that
\begin{equation}
\mathbf{P}_{u}(w)=\mathbf{P}_{\perp}(u\cdot\mathfrak{\mathring{d}%
}w),~~~\mathbf{P}_{w}(u)=\mathbf{P}_{\perp}(w\cdot\mathfrak{\mathring{d}}u).
\label{ap33}%
\end{equation}

Now, we already know from Eq.(\ref{O2}) that $\mathbf{P}_{\perp}%
(u\cdot\mathfrak{\mathring{d}}w)$ and $\mathbf{P}_{\perp}(w\cdot
\mathfrak{\mathring{d}}u)$ \ are equal and thus
\begin{equation}
\mathbf{P}_{u}(w)=\mathbf{P}_{w}(u). \label{ap4}%
\end{equation}
From this equation also follows immediately that%
\begin{equation}
\mathbf{P}_{u}\mathbf{P}(w)=\mathbf{P}_{w}\mathbf{P}(u). \label{ap5}%
\end{equation}

Now given that each $\mathcal{X}\in\sec\mathcal{C}\ell(\mathring
{M},\boldsymbol{\mathring{g}})$ can be written as $\mathcal{X}=\mathcal{X}%
_{\parallel}+\mathcal{X}_{\perp}$, with $\mathcal{X}_{\parallel}%
=\mathbf{P}(\mathcal{X)}$ we get from Eq.(\ref{ap3}) that
\begin{equation}
\mathbf{PP}_{u}(\mathcal{X}_{\parallel})=\mathbf{0}\text{,~~~~}\mathbf{PP}%
_{u}(\mathcal{X}_{\perp})=\mathbf{P}_{u}(\mathcal{X}_{\perp}). \label{ap6}%
\end{equation}
Also, from Eq.(\ref{ap2}) we have immediately taking into account
Eq.(\ref{ap6}) for any $\mathcal{C},\mathcal{D}\in\sec\mathcal{C}%
\ell(\mathring{M},\boldsymbol{\mathring{g}})$ that
\begin{equation}
\mathbf{P}_{u}(\mathcal{C}_{\parallel}\wedge\mathcal{D}_{\perp})=\mathcal{C}%
_{\parallel}\wedge\mathbf{P}_{u}(\mathcal{D}_{\perp}). \label{ap7}%
\end{equation}

\end{document}